\documentclass[11pt]{article}

\DeclareMathAlphabet{\mathrmbf}{\encodingdefault}{\itdefault}{bx}{n}

\usepackage{comment}
\usepackage{parskip}
\usepackage{times}
\usepackage{fullpage}
\usepackage{amsmath}
\usepackage{amsfonts}
\usepackage{amssymb}
\usepackage{graphicx}
\usepackage{color}
\usepackage{url,doi}
\usepackage{hyperref}
\usepackage{bbold}
\usepackage{xspace}
\usepackage{authblk}

\newtheorem{theorem}{Theorem}
\newtheorem{observation}[theorem]{Observation}
\newtheorem{lemma}[theorem]{Lemma}

\newtheorem{defn}{Definition}

\newtheorem{protocol}{Protocol \boldmath}

\newcommand{\op}[1]{\mathsf{#1}}
\newcommand{\party}[1]{\mathbb{#1}}
\newcommand{\ket}[1]{\ensuremath{|{#1}\rangle}}
\newcommand{\tinner}[3]{\ensuremath{\langle{#1}|#2|{#3}\rangle}}
\newcommand{\bra}[1]{\ensuremath{\langle{#1}|}}
\newcommand{\Tr}{\mbox{\rm Tr}}
\newcommand{\tinyspace}{\mspace{1mu}}
\newcommand{\abs}[1]{\left\lvert\tinyspace #1 \tinyspace\right\rvert}

\newcommand{\diag}{\mathop{\mathrm{diag}}}

\newcommand{\head}{\textup{head}}
\newcommand{\tail}{\textup{tail}}
\newcommand{\psd}{\text{positive semidefinite}}

\newcommand{\nth}[1]{\ensuremath{#1}^{\textup{th}}}

\newcommand{\id}{{\mathbb{1}}}
\newcommand{\dualm}{\mathcal{M}}

\newcommand{\zo}{\{0,1\}}
\def\squareforqed{\hbox{\rlap{$\sqcap$}$\sqcup$}}
\def\qed{\ifmmode\squareforqed\else{\unskip\nobreak\hfil
\penalty50\hskip1em\null\nobreak\hfil\squareforqed
\parfillskip=0pt\finalhyphendemerits=0\endgraf}\fi}

\newenvironment{proof}{\begin{trivlist}\item[]{\flushleft\bf Proof }}
{\qed\end{trivlist}}

\newenvironment{proofof}[1]{\begin{trivlist}\item[]{\flushleft\bf 
Proof of~#1 }}
{\qed\end{trivlist}}

\makeatletter
\newif\if@borderstar
\def\bordermatrix{\@ifnextchar*{%
  \@borderstartrue\@bordermatrix@i}{\@borderstarfalse\@bordermatrix@i*}%
}
\def\@bordermatrix@i*{\@ifnextchar[{%
  \@bordermatrix@ii}{\@bordermatrix@ii[()]}
}
\def\@bordermatrix@ii[#1]#2{%
  \begingroup
    \m@th\@tempdima8.75\p@\setbox\z@\vbox{%
      \def\cr{\crcr\noalign{\kern 2\p@\global\let\cr\endline }}%
      \ialign {$##$\hfil\kern 2\p@\kern\@tempdima & \thinspace %
      \hfil $##$\hfil && \quad\hfil $##$\hfil\crcr\omit\strut %
      \hfil\crcr\noalign{\kern -\baselineskip}#2\crcr\omit %
      \strut\cr}}%
    \setbox\tw@\vbox{\unvcopy\z@\global\setbox\@ne\lastbox}%
    \setbox\tw@\hbox{\unhbox\@ne\unskip\global\setbox\@ne\lastbox}%
    \setbox\tw@\hbox{%
      $\kern\wd\@ne\kern -\@tempdima\left\@firstoftwo#1%
        \if@borderstar\kern2pt\else\kern -\wd\@ne\fi%
      \global\setbox\@ne\vbox{\box\@ne\if@borderstar\else\kern 2\p@\fi}%
      \vcenter{\if@borderstar\else\kern -\ht\@ne\fi%
        \unvbox\z@\kern-\if@borderstar2\fi\baselineskip}%
        \if@borderstar\kern-2\@tempdima\kern2\p@\else\,\fi\right\@secondoftwo#1 $%
    }\null \;\vbox{\kern\ht\@ne\box\tw@}%
  \endgroup
}
\makeatother

\newcommand{\bfunc}[1]{\textsf{#1}}
\begin{document}
\title{Quantum Nonlocal Boxes Exhibit Stronger Distillability}

\renewcommand\Authands{ and }
\renewcommand\Affilfont{\slshape\small}

\author{Peter H{\o}yer}
\author{Jibran Rashid}

\affil{\,Department of Computer Science, University of Calgary \authorcr
\textsl{2500 University Drive N.W., Calgary, AB, T2N 1N4~Canada.}
\textup{\small \{hoyer,\,jrashid\}@ucalgary.ca}}

\date{}

\maketitle

\begin{abstract}
The hypothetical nonlocal box (\bfunc{NLB}) proposed by Popescu and Rohrlich allows two spatially separated parties, Alice and Bob, to exhibit stronger than quantum correlations. If the generated correlations are weak, they can sometimes be distilled into a stronger correlation by repeated applications of the \bfunc{NLB}. Motivated by the limited distillability of \bfunc{NLB}s, we initiate here a study of the distillation of correlations for nonlocal boxes that output quantum states rather than classical bits (\bfunc{qNLB}s). We propose a new protocol for distillation and show that it asymptotically distills a class of correlated quantum nonlocal boxes to the value $\frac{1}{2}(3\sqrt{3}+1) \approx 3.098076$, whereas in contrast, the optimal non-adaptive parity protocol for classical nonlocal boxes asymptotically distills only to the value~$3.0$. We~show that our protocol is an optimal non-adaptive protocol for~$1$, $2$ and~$3$ \bfunc{qNLB} copies by constructing a matching dual solution for the associated primal semidefinite program (SDP). We~conclude that \bfunc{qNLB}s are a stronger resource for nonlocality than \bfunc{NLB}s. The main premise that develops from this conclusion is that the \bfunc{NLB} model is not the strongest resource to investigate the fundamental principles that limit quantum nonlocality. As such, our work provides strong motivation to reconsider the status quo of the principles that are known to limit nonlocal correlations under the framework of \bfunc{qNLB}s rather than \bfunc{NLB}s.
\end{abstract}

\section{Nonlocality distillation}
\label{sec:intro}
Consider two parties, Alice and Bob, spatially separated and isolated, interested in jointly computing some boolean function $f(\cdot,\cdot)$.  A~third party, David, provides Alice with an input~$x$ (unbeknown to Bob) and Bob with an input $y$ (unbeknown to Alice) and challenges them to compute the bit \mbox{$f(x,y)$}. David allows Alice and Bob to communicate, but charges for each and every bit communicated between them. Alice and Bob therefore pre-agree upon a protocol that minimizes the amount of communication required for them to compute the bit~$f(x,y)$. This is what we know as communication complexity~\cite{Kush}.

It seems entirely impossible to jointly compute a non-trivial function if no information can be interchanged between Alice and Bob, and it is indeed one of the first results typically shown in any introduction to communication complexity.  But as soon as one tweaks the models ever so slightly, surprising results are possible. The nonlocal box is one such tweaking. 

A~\emph{nonlocal box} (\bfunc{NLB}) is a device shared between two parties that, in itself is incapable of transferring any information from Alice to Bob, or vice-versa. A~nonlocal box takes two bits as input, a bit $x$ from Alice and a bit $y$ from Bob, and outputs two bits, a bit $a$ provided to Alice (and only Alice) and a bit $b$ provided to Bob (and only Bob).  If the two input bits $x$ and $y$ from Alice and Bob equal $(0,0)$, $(0,1)$, or $(1,0)$, the box (by definition) provides Alice and Bob with identical bits. That is, either both of them receive $0$ or both of them receive~$1$, each case happening with probability~$\frac{1}{2}$.  If the two parties both give the box a $1$ as input, the box provides Alice and Bob with opposite bits $x$ and $y$, again each of the two cases $01$ and $10$ happening with equal probabilities~$\frac{1}{2}$. See Figure~\ref{fig:nlb}. 

The correlations related to the \bfunc{NLB} and some of their key properties were initially discovered by Khalfin and Csirel'son~\cite{Cirel85} in 1985. Reintroduced by Popescu and Rohrlich in their seminal 1994 paper~\cite{Popescu94b}, \bfunc{NLB}s, have since undergone extensive scrutiny.
\begin{figure}[t]
\centering
\scalebox{1}{\includegraphics{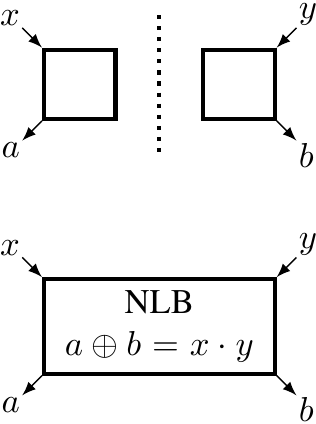}}
\caption[A Nonlocal Box (\bfunc{NLB})]{A nonlocal box. The figures depict the \bfunc{NLB} as a single box rather than two invidual boxes shared between Alice and Bob.}
\label{fig:nlb}
\end{figure}
A nonlocal box is by definition \emph{non-signalling}.  The marginal of the bit $a$ received by Alice is uniform irrespective of whether Bob inputs $0$ or $1$ to the box, and she thus does not obtain any information about Bob's input.  Yet the parity of the two output bits $a$ and $b$ is perfectly correlated with the logical $\textsf{AND}$ of the two input bits $x$ and~$y$. 

The perfect nonlocal box (as defined above) is powerful enough to render all of communication complexity trivial, i.e., any boolean function may be computed by a single bit of communication between Alice and Bob~\cite{Dam05}.  Even if we modify the box so that, for each of the four possible inputs, it provides an output of the expected parity only with probability at least $\frac{3+\sqrt{6}}{6} \approx 0.908$, it would still be possible to compute any boolean function with bounded error using only a single bit of
communication~\cite{Brassard05}!

Our motivation for the current work develops from a simple open question, i.e., do noisy \bfunc{NLB}s within the range \mbox{$\cos^2 \left( \frac{\pi}{8} \right) < p < \frac{3+\sqrt{6}}{6}$} allow for trivial communication complexity? The agenda in this approach is to show that quantum mechanics restricts correlation sources that result in a world in which surprisingly powerful information processing procedures could be performed. A different example of this line of work from cryptography is due to Buhrman et al.~\cite{Buhrman06}, which shows that \bfunc{NLB}s can be used to be perform any two-party secure computation. They build protocols for bit commitment and oblivious transfer using \bfunc{NLB}s, both of which are known to be impossible to achieve using quantum mechanics.

Nonlocality distillation refers to the extent by which we can turn
weak nonlocal boxes into more pure nonlocal boxes through a protocol. The idea is to consider whether it is possible for the players to concentrate the nonlocality in $n$ copies of an imperfect nonlocal source to form a stronger nonlocal correlation source. In this sense it may be considered similar to entanglement distillation. We have gained some understanding of when nonlocality can be distilled~\cite{SWolf08a, SWolf08b, SWolf09, Brunner09, Allcock09a, Allcock09b, Hoyer10, Forster11}, when it cannot~\cite{Short09} and when it appears in bound form~\cite{Brunner10}.  In~general, the results suggest that distillation is only possible under special favorable circumstances and that large classes of nonlocal boxes are not distillable. 

The apparent limited distillability of \bfunc{NLB}s even under adaptive protocols seems to suggest that distillation may not be a strong enough framework to draw conclusions regarding limits on nonlocal correlations. With the introduction of the \bfunc{qNLB} model, we hope to change the situation. Historically, it took a period of more than half a century to realize that a more feasible interpretation of Bell inequality violations is to view them as a resource for processing information, rather than as paradoxes. Apparently, the same restrictive reasoning haunts us where we view the violation of Csirelson's inequality as something to be written off as an impossibility. No doubt, it is crucial to determine the principles that determine bounds on quantum correlations, however, another approach is to construct communication models that produce exactly the correlations within the no-signalling polytope.

We~approach stronger that quantum correlations with this new perspective. Rather than
considering a hypothetical box resource, the spatially separated parties Alice and Bob, are now provided access to a trusted third party Charlie. Charlie is allowed to communicate with Alice and Bob without allowing communication \emph{between} Alice and Bob. Consider the scenario depicted in Figure~\ref{fig:physmodel}.  David who wants to compute a boolean function $f(x,y)$, provides Alice and Bob with a description of~$f$ and the partitioned inputs~$x$ and~$y$. Alice and Bob may now use another trusted party Charlie who simulates the actions of a \bfunc{NLB}/\bfunc{qNLB}.  This allows Alice and Bob to determine and transmit~$a$ and~$b$ to David such that \mbox{$f(x,y)=a \oplus b$}. The three parties are able to help in computing the function~$f$ without any of them having access to complete information.  Alice and Bob know the function, but not the complete input, nor its value, while Charlie knows the input without knowing the function being computed.  

\begin{figure}[t]
\centering
\scalebox{1.3}{\includegraphics{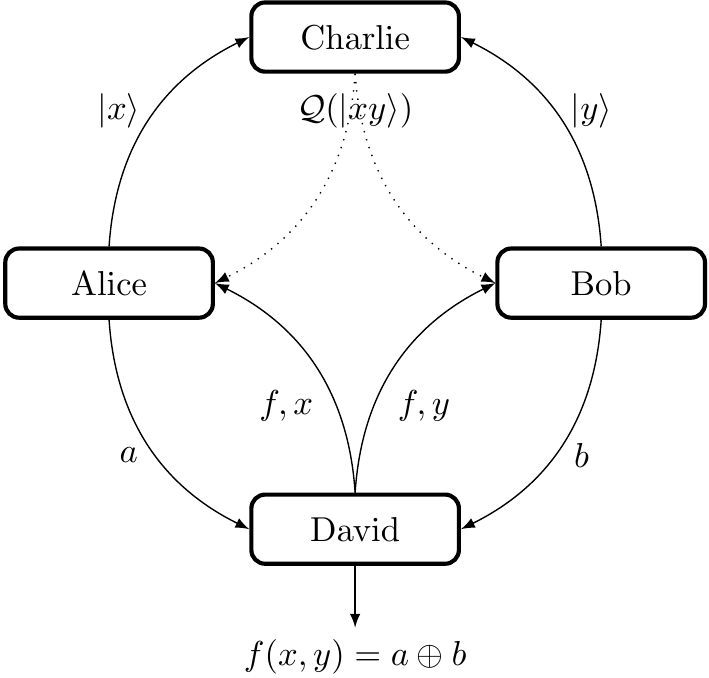}} 
\caption[Possible Physical Simulation of a \bfunc{\bfunc{qNLB}}]{A~possible physical realization of the quantum nonlocal box. David who wants to compute a boolean function $f(x,y)$ transmits the description of~$f$ as well as~$x$ and~$y$ to Alice and Bob respectively. Alice and Bob share the state $\ket{xy}$ with Charlie, who implements and distributes the results of the no-signalling map $\mathcal{Q}(\ket{xy})$. Finally, Alice and Bob
  send~$a$ and~$b$ to David, such that \mbox{$f(x,y) = a \oplus b.$}}
\label{fig:physmodel}
\end{figure}

Charlie's actions can be modelled by a nonlocal box that produces correlated physical systems as output. A~quantum nonlocal box, abbreviated \bfunc{qNLB}, takes as input a joint quantum state and outputs a joint quantum state.  A~priori, such a model may not obey our non-signalling requirement since any unitary $\op{U}_{AB}$ not on the form $\op{U}_{A} \otimes \op{U}_{B}$ allows for signalling~\cite{Benn03, Piani06}.  It~thus may appear that a quantum generalization of the \bfunc{NLB} model would always allow for signalling, but this only holds true if we restrict the maps to be unitary.  Quantum nonlocal boxes that satisfy the non-signalling
requirement and allow for quantum states as output are possible when
we drop the requirement of the box being unitary.  Such boxes have
previously been studied under the notion of causal maps, completely
positive trace-preserving maps, and non-signalling
operations~\cite{Marcovitch07, Piani06, Beckman01, Gus10, Buhrman04}.

As~our main result, we show that \bfunc{qNLB}s exhibit strictly stronger nonlocality distillation than \bfunc{NLB}s when restricted to non-adaptive distillation protocols. We show that in such
a scenario, the optimal non-adaptive nonlocality distillation protocol for Alice and Bob
asymptotically performs better than the optimal non-adaptive distillation parity protocol
for \bfunc{NLB}s~\cite{Hoyer10}.
\begin{theorem} 
\label{thm:mainthm}
Quantum nonlocal boxes exhibit stronger nonlocality distillation for
non-adaptive protocols than the optimal non-adaptive parity protocol
for classical nonlocal boxes.
\end{theorem}
We~prove our main theorem by setting up a
semidefinite programming framework~\cite{Wehner05d} for analyzing
non-adaptive protocols for \bfunc{qNLB} distillation.  We then use this
framework to define and give a protocol for \bfunc{qNLB} distillation and show
that it outperforms the optimal non-adaptive protocol for classical nonlocal boxes~\cite{Hoyer10}. We~show that our protocol is an optimal non-adaptive protocol for the class of correlated \bfunc{qNLB}s,
given~$1$,~$2$ and~$3$ copies by constructing a dual solution that attains the same value as the primal.

\section{Distillation protocols}

We define the \emph{value} of a nonlocal box as the sum of the biases
that the parity of the box agrees with the logical \textsf{and} of the
input bits, over all four possible inputs,
\begin{equation}\label{eq:NLBvalue}
V =  \sum_{x,y \in \zo} \textup{Pr}[ a \oplus b = x \cdot y]
   - \sum_{x,y \in \zo} \textup{Pr}[ a \oplus b \neq x \cdot y].
\end{equation}
Brunner and Skrzypczyk considered and analyzed in~\cite{Brunner09} a
class of \bfunc{NLB}s that has only one-sided errors and labelled them
correlated \bfunc{NLB}s.
\begin{defn} \label{def:nlb} 
A~\emph{correlated} \bfunc{NLB} maps the three inputs $00$, $01$ and $10$ to
the output $00$ with probability $\frac{1}{2}$, and to the output $11$
with complementary probability $\frac{1}{2}$.  It maps the input $11$
to either of the two outputs $01$ and $10$ with equal probabilities
$\frac{p}{2}$, and to either of the two outputs $00$ and $11$ with
equal probabilities $\frac{1-p}{2}=\frac{q}{2}$.  Here $p\in [0,1]$ denotes the probability that, on input $11$, the output of the \bfunc{NLB} is of odd parity. Similarly, $q=1-p$ denotes denotes the probability that, on input $11$, the output of the \bfunc{NLB} is of even parity.
\end{defn}
The value of a correlated box is $3+p-(1-p) = 2(1+p)$, and the value of a perfect \bfunc{NLB} is~$4$.

Consider now that Alice and Bob share $n$ instances of a correlated
nonlocal box, all with the same parameter~$p$.  Their goal is to
simulate the behaviour of a correlated nonlocal box with a better
parameter \mbox{$p'>p$} by using some pre-agreed upon protocol.  They
may use the $n$ \bfunc{NLB} instances as well as shared randomness, but are
not allowed to communicate.  If their protocol achieves a higher value
$p'$ than $p$, we call the protocol a \emph{distillation protocol}. A~distillation protocol using $n$ nonlocal boxes is said to be \emph{non-adaptive} if Alice is required to provide her input $x$ to
all $n$ nonlocal boxes and Bob is required to provide his input $y$ to
all $n$ boxes. 

A~correlated \bfunc{NLB} can be asymptotically distilled to a perfect \bfunc{NLB} by
an adaptive protocol~\cite{SWolf09} as follows.  Consider a single
execution of a correlated \bfunc{NLB} with input bits $x$ and~$y$ and output
bits $a$ and~$b$.  If the two inputs are both~$1$, a correlated \bfunc{NLB}
may output an incorrect correlation, whereas, if at least one of the
two inputs is~$0$, the output is always correctly of even parity.
Viewed from the perspective of the output bits $a$ and~$b$, if the
parity $a\oplus b$ is odd, we can conclude that the two inputs bit
were both~$1$, and that the output therefore is correct.  Only if the
output $a \oplus b$ is even, can we \emph{not} conclude with certainty
that the output is correct.  An~adaptive protocol can use this
one-sidedness of error to distill to the asymptotically optimal value
of~$4$ by patiently waiting till the first time a usage of the
correlated \bfunc{NLB} yields an output of odd parity.  This can be detected
distributively (but not locally), and once detected, the protocol
adaptively (and distributively) adjusts further usages of the \bfunc{NLB}s so
that all future outputs are of even parity.  The eventual distributive
detection of an output of odd parity reveals that the input bits were
both~$1$, and the lack of an output pair of odd parity indicates that
at least one of the two input bits were~$0$.  An odd parity output
will eventually occur, allowing us to asymptotically distill to the
optimal value~$4$.

A~non-adaptive protocol can in contrast not distill to the value~$4$.
By~not allowing for adaptiveness, the distributive detection of the
parity of the output can not be fed back into the system, and the
protocol then fails in taking full advantage of the knowledge it
possesses.  A~non-adaptive protocol must patiently wait till all
outputs are produced, at which stage its best strategy is to take the
parity of a certain number $k$ of its outputs~\cite{Hoyer10}.

\begin{theorem}[\cite{Hoyer10}] \label{thm:parityopt}
The value attainable by any non-adaptive protocol using at most $n$
correlated \bfunc{NLB}s is upper bounded by
\begin{equation}
V = \begin{cases} 
3 - (q-p)^n & \text{ if $0 \leqslant p < \frac{1}{2}$}\\
2(1+p) & \text{ if $\frac{1}{2} \leqslant p \leqslant 1$,}
\end{cases}
\end{equation}
and this value is attainable by the parity protocol.
\end{theorem}

In~this work, we consider the case that the \bfunc{NLB}s take quantum states
as input and produce quantum states as output. See Figure~\ref{fig:qnlb}.
\begin{defn}[Quantum nonlocal box]
\label{def:qnosigbox}
A~\emph{quantum nonlocal box} (\bfunc{qNLB}) $\mathcal{Q}$ takes as input a
product state $\ket{\psi_{xy}} \in
\{\ket{00},\ket{01},\ket{10},\ket{11}\} \in \mathcal{H}_\party{A}
\otimes \mathcal{H}_{\party{B}}$ and outputs a state \mbox{$\rho_{xy}
  \in \mathcal{H}_\party{A} \otimes \mathcal{H}_\party{B}$} such that
for every map \mbox{$\Gamma_{\party{A}}: \mathcal{H}_\party{A} \mapsto
  \mathcal{H}_\party{A}$} and \mbox{$\Gamma_{\party{B}}:
  \mathcal{H}_\party{B} \mapsto \mathcal{H}_\party{B}$} the following
two no-signalling conditions hold,
\begin{align*}
\Tr_\party{A}  \mathcal{Q} 
\big( \left(\Gamma_{\party{A}} \otimes \id \right) 
\ket{\psi_{xy}}\bra{\psi_{xy}} \big)
 &= \Tr_\party{A} \mathcal{Q} \ket{\psi_{xy}}\bra{\psi_{xy}} \\
\Tr_\party{B} \mathcal{Q} 
\big( \left(\id \otimes \Gamma_{\party{B}} \right) 
\ket{\psi_{xy}}\bra{\psi_{xy}} \big)
 &= \Tr_\party{B} \mathcal{Q} \ket{\psi_{xy}}\bra{\psi_{xy}}.
\end{align*}
\end{defn}

\begin{figure}
\centering
\scalebox{1}{\includegraphics{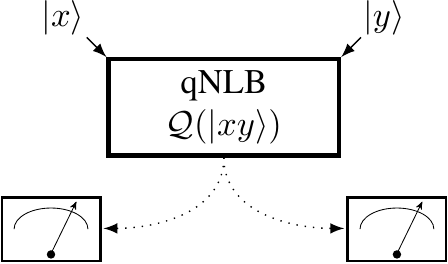}}
\caption[figure1]{A quantum nonlocal box}
\label{fig:qnlb}
\end{figure}
In~particular, we consider the class of correlated \bfunc{qNLB} that
generalizes the class of correlated \bfunc{NLB}s.
\begin{defn} \label{def:qnlb} 
A~\emph{correlated} \bfunc{qNLB} maps the three inputs $\ket{00}$, $\ket{01}$
and $\ket{10}$ to the pure state
$\ket{\psi} = \frac{1}{\sqrt{2}}\big(\ket{00} + \ket{11} \big)$,
and maps the input $\ket{11}$ to the mixed state
$\rho = p\ket{\phi}\bra{\phi} + q \ket{\psi}\bra{\psi}$,
where $\ket{\phi} =\frac{1}{\sqrt{2}}\left(\ket{01} + \ket{10}
\right)$ is a superposition over the two odd-parity states, $p \in
      [0,1]$ a probability, and $q = 1-p$ the complementary
      probability.
\end{defn}

Given that Alice and Bob share~$n$ copies of a correlated \bfunc{qNLB} and
measure observables~$\op{A}_x$ and~$\op{B}_y$ with eigenvalues $\pm1$
for input bits~$x$ and~$y$, respectively, the value attained for the
CHSH inequality~\cite{CHSH69}~is
\begin{equation}
\label{v:qnlb}
V = \bra{\psi}^{\otimes n}( \op{A}_0 \otimes \op{B}_0 + \op{A}_0 \otimes \op{B}_1
+ \op{A}_1 \otimes \op{B}_0) \ket{\psi}^{\otimes n} - \Tr(\op{A}_1 \otimes
\op{B}_1 \rho^{\otimes n}).
\end{equation}
A~\bfunc{qNLB} is at least as powerful as an \bfunc{NLB}: For any value of $p$, Alice
and Bob can use a correlated \bfunc{qNLB} to simulate the correlation of a
correlated \bfunc{NLB} by simply measuring each of their outputs in the
computational basis.  In this paper, we formally prove that \bfunc{qNLB}s are
strictly more powerful in extracting nonlocality than are \bfunc{NLB}s. See Figure~\ref{fig:qnlbprotocol} for the structure of non-adaptive distillation protocol for \bfunc{qNLB}s. We~establish our main Theorem~\ref{thm:mainthm} by giving an explicit non-adaptive protocol that attains a higher distilled value for correlated \bfunc{qNLB}s than the optimal parity protocol attains for correlated \bfunc{NLB}s.  The amount of distillability achievable by non-adaptive protocols for \bfunc{NLB}s is characterized in~\cite{Hoyer10}, here specialized to correlated \bfunc{NLB}s as Theorem~\ref{thm:parityopt} above.

\begin{figure}
\centering
\scalebox{1}{\includegraphics{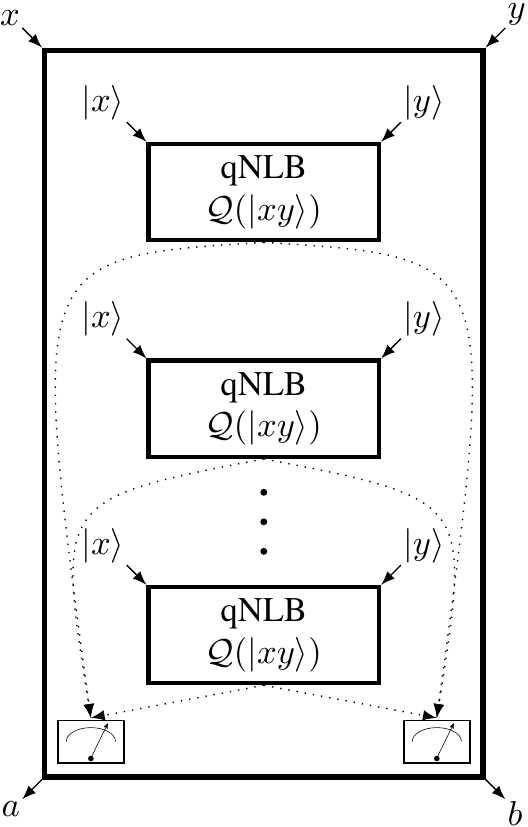}}
\caption[figure1]{A non-adaptive distillation protocol for \bfunc{qNLB}s.}
\label{fig:qnlbprotocol}
\end{figure}

\begin{table}
\begin{equation*}
\begin{array}{c|c|c}
& \text{\bfunc{NLB} distill?} 
& \text{Value}\\\hline
p=0 
    & \text{no} 
    & 2 \\
0<p<\frac{1}{2} 
    & \text{yes} 
    & 3-(q-p)^n \\
\frac{1}{2} \leqslant p \leqslant 1 
    & \text{no}
    & 2(1+p) \\
\end{array}
\end{equation*}
\caption{Non-adaptive distillation of correlated \bfunc{NLB}s is possible if
  and only if $0 < p < \frac{1}{2}$, for which they can be
  asymptotically distilled to the value~3.}
\label{table:NLBdistillsummary}
\end{table}

\begin{table}
\begin{equation*}
\begin{array}{c|c|c}
& \text{\bfunc{qNLB} distill?} 
& \text{Value} \\\hline
p=0 
    & \text{no} 
    & 2 \\
0<p<\frac{1}{2} 
    & \text{yes} 
    & (3+(q-p)^n)\cos(\phi) +\frac{1}{2}(1-(q-p)^n) \\
p=\frac{1}{2} 
    & \text{no for $n\leqslant 3$} 
    & \frac{1}{2}(3\sqrt{3}+1) \\
\frac{1}{2}<p < \frac{2}{3} 
    & \text{no for $n\leqslant 3$} 
    & 3 \cos(\phi) - q\cos(3\phi) + p \\
\frac{2}{3} \leqslant p < 1 
    & \text{no for $n\leqslant 3$} 
    & 2(1+p)\\
p=1
    & \text{no}
    & 4
\end{array}
\end{equation*}
\caption{Non-adaptive distillation of correlated \bfunc{qNLB}s is possible
  when $0 < p < \frac{1}{2}$, for which they can be asymptotically
  distilled to the value~$\frac{1}{2}(3\sqrt{3}+1) \approx 3.098076$.
  When $\frac{1}{2} \leqslant p < 1$, non-adaptive distillation is not
  possible using at most 3 \bfunc{qNLB}s.  Measurement angle~$\phi$ depends
  on~$p$ and is defined in Eq.~\ref{eq:phi}.}
\label{table:qNLBdistillsummary}
\end{table}

We summarize the known results on non-adaptive distillation of
classical and quantum correlated nonlocal boxes in Tables
\ref{table:NLBdistillsummary} and~\ref{table:qNLBdistillsummary}.  For
\bfunc{NLB}s, we have complete knowledge: correlated \bfunc{NLB}s are non-adaptively
distillable if and only if $0<p<\frac{1}{2}$, and the parity protocol
of Forster \emph{et al.}~\cite{Forster11} is an optimal non-adaptive
protocol~\cite{Hoyer10}.  For \bfunc{qNLB}s, we show here that correlated
\bfunc{qNLB}s are non-adaptively distillable when $0<p<\frac{1}{2}$, and that
correlated \bfunc{qNLB}s can \emph{not} be non-adaptively distilled when
$\frac{1}{2} \leqslant p < 1$ if we Alice and Bob are allowed to use
at most 3 \bfunc{qNLB}s.  When $\frac{1}{2} \leqslant p < 1$, we show that the
single-usage \bfunc{qNLB} protocol of Piani \emph{et al.}~\cite{Piani06} is
optimal among all non-adaptive protocols using at most 3 \bfunc{qNLB}s.

The values attainable are plotted in Figure~\ref{fig:distillsummary}.
When $0 < p < \frac{2}{3}$, \bfunc{qNLB}s achieves a strictly larger value
than \bfunc{NLB}s for any fixed value of~$n$.  For $0 < p < \frac{1}{2}$,
\bfunc{qNLB}s can be asymptotically distilled to $\frac{1}{2}(3\sqrt{3}+1)
\approx 3.098076$, whereas \bfunc{NLB}s can only be asymptotically distilled to
the value~3.

\begin{figure}
\centering
\scalebox{.9}{\includegraphics{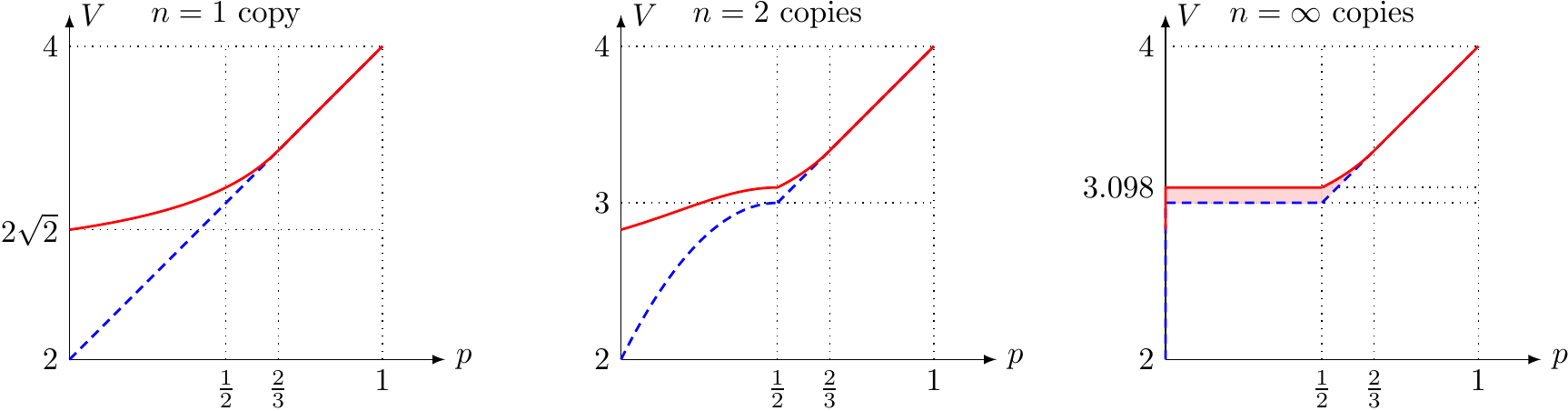}}
\caption{Value attained by our Protocol~$\mathcal P$ for correlated \bfunc{qNLB}s (solid line) and the parity protocol for \bfunc{NLB}s (dotted lines)~\cite{Forster11}. For $0 < p < \frac{1}{2}$, parity distills to $3.0$, while our protocol distills to $\frac{1}{2}\Big(3\sqrt{3}+1 \Big) \approx 3.098$. For $\frac{1}{2}  \leqslant  p  <  \frac{2}{3}$, even though \bfunc{qNLB}s attain a higher value than \bfunc{NLB}s, no distillation occurs. Finally, for $\frac{2}{3}  \leqslant  p  \leqslant  1$, both the protocols attain the same value without any distillation taking place.}
\label{fig:distillsummary}
\end{figure}

\section{Our distillation protocol}
\label{sec:theprotocol}
We propose the following protocol~\ref{protocolP} for non-adaptively
distilling correlated \bfunc{qNLB}s.

\begin{protocol}
\label{protocolP}
Let Alice and Bob share~$n$ identical copies of a correlated \bfunc{qNLB} of
parameter $p$ and let them receive input bits~$x$ and~$y$,
respectively. Their observables~$\op{A}_x$ and~$\op{B}_y$ are given~by
\begin{align*} 
\op{A}_x &= \cos \left(\frac{\phi}{2} +x\phi\right)\op{Z}
     +(-1)^x\sin \left(\frac{\phi}{2} +x\phi\right)\op{X}\\
\op{B}_y &= \cos \left(\frac{\phi}{2} +y\phi\right)\op{Z}
     -(-1)^y\sin \left(\frac{\phi}{2}+y\phi \right)\op{X}.
\end{align*}
The operators~$\op{Z}$ and~$\op{X}$ for the observables~$\op{A}_x$
and~$\op{B}_y$ are chosen based on the value of~$p$,
\begin{equation}
\label{eq:observables}
\op{Z} = \left\{\begin{array}{ll}
\sigma_\op{z}^{\otimes n} 
   & \text{ if } 0 < p < \frac{1}{2} \\
\sigma_\op{z} \otimes \id^{\otimes n-1} 
   & \text{ if } \frac{1}{2} \leqslant p \leqslant 1 
\end{array} \right. 
\quad \text{and} \quad
\op{X} = \left\{\begin{array}{ll}
\sigma_\op{x}^{\otimes n} 
     & \text{ if } 0 < p < \frac{1}{2} \\
\sigma_\op{x} \otimes \id^{\otimes n-1} 
     & \text{ if } \frac{1}{2} \leqslant p \leqslant 1.
\end{array} \right.
\end{equation}
The measurement angle~$\phi$ depends on~$p$ and is chosen such that it
maximizes the value attained for the CHSH inequality,
\begin{equation}
\label{eq:phi}
\cos^2 \left(\phi(p)\right) = \left\{ 
\begin{array}{ll}
\frac{1}{4}\left( \frac{3+(q-p)^n}{1+(q-p)^n} \right) 
  & \text{ if } 0 < p < \frac{1}{2} \\
\frac{1+q}{4q} 
  & \text{ if } \frac{1}{2} \leqslant p < \frac{2}{3}\\
1 & \text{ if } \frac{2}{3} \leqslant p \leqslant 1.\\
\end{array} \right.
\end{equation}
\end{protocol}

The observables chosen by Alice and Bob in Eq.~\ref{eq:observables} in
Protocol~\ref{protocolP} depends on the probability~$p$.  If~$0 < p <
\frac{1}{2}$, Alice and Bob non-trivially use all $n$ available \bfunc{qNLB}s.
They view those $n$ \bfunc{qNLB}s as a single \bfunc{qNLB} and each applies an
observable given by global measurement angle in a two dimensional
space spanned by the two observables $\sigma_\op{z}^{\otimes n}$ and
$\sigma_\op{x}^{\otimes n}$ (see Eq.~\ref{eq:observables}).  Thus viewed as
a two-dimensional rotation, our chosen observables can be seen as a
generalization of the measurements in the protocol of Piani et al.~\cite{Piani06} for a single \bfunc{qNLB}.

If~$\frac{1}{2} \leqslant p \leqslant 1$, Alice and Bob effectively
choose to use only a single \bfunc{qNLB} by applying the identity observable
$\id = \left(\begin{smallmatrix}1 & 0 \\ 0 &1\end{smallmatrix}\right)$
  on all but the first \bfunc{qNLB}.  The output bits $a$ and $b$ of Alice and
  Bob depend only on the output bits of the first \bfunc{qNLB}.  Therefore, an
  alternative protocol achiving the same values as ours, can be
  constructed in which Alice and Bob first choose to use a number $k$
  of \bfunc{qNLB}s, discard the remaining $n-k$ \bfunc{qNLB}s, and then each apply an
  observable on the $k$ selected \bfunc{qNLB}s as in Protocol~\ref{protocolP}.
  When $0 < p < \frac{1}{2}$, they pick \mbox{$k=n$}, and when $p
  \geqslant \frac{1}{2}$, they pick \mbox{$k=1$}.

Having specified the four observables $\op{A}_0$, $\op{A}_1$,
$\op{B}_0$, and $\op{B}_1$, we compute the value attained by our
Protocol~\ref{protocolP} by plugging into Eq.~\ref{v:qnlb}.
\begin{lemma} 
\label{lm:valueofprotocol}
Protocol~\ref{protocolP} attains the value
\begin{equation}
\label{eq:valueofprotocolP}
V =
\begin{cases}
(3+(q-p)^n)\cos(\phi) +\frac{1}{2}(1-(q-p)^n)
& \text{ if $0 < p < \frac{1}{2}$}\\
2(1+q)\cos(\phi) +p 
& \text{ if $\frac{1}{2} \leqslant p < \frac{2}{3}$}\\
2(1+p) 
& \text{ if $\frac{2}{3} \leqslant p \leqslant 1$.}
\end{cases}
\end{equation}
\end{lemma}

A~complete proof of Lemma~\ref{lm:valueofprotocol} is given in
Appendix~\ref{appendix:protocolvalue}.  In~the next two sections, we
prove that our protocol is an optimal non-adaptive protocol for any
$0<p\leqslant 1$ and any \mbox{$n \leqslant 3$}.

\section{Protocol~\ref{protocolP} is optimal for a single copy}
\label{sec:singlecopy}

We~now show that no other protocol can achive a higher value $V$ than
ours when Alice and Bob are given a single \bfunc{qNLB}.  When $n=1$, the
expression for the value $V$ given in Eq.~\ref{v:qnlb} simplifies to
\begin{equation*}
  \tinner{\psi}{\op{A}_0 \otimes \op{B}_0}{\psi}
+ \tinner{\psi}{\op{A}_0 \otimes \op{B}_1}{\psi}
+ \tinner{\psi}{\op{A}_1 \otimes \op{B}_0}{\psi}
- p \tinner{\phi}{\op{A}_1 \otimes \op{B}_1}{\phi}
- q \tinner{\psi}{\op{A}_1 \otimes \op{B}_1}{\psi}.
\end{equation*}
Using that the two Bell states $\ket{\psi}$ and $\ket{\phi}$ (given in
Definition~\ref{def:qnlb}) can be locally mapped to each other,
$\ket{\psi} = (\id \otimes \sigma_\op{x})\ket{\phi}$,
we rewrite the optimization problem in terms of a single state
$\ket{\psi}$,
\begin{equation}
\label{eq:optsinglecopy}
V = 
  \tinner{\psi}{\op{A}_0 \otimes \op{B}_0
+   {\op{A}_0 \otimes \op{B}_1}
+   {\op{A}_1 \otimes \op{B}_0}
- p (\op{A}_1 \otimes \sigma_\op{x}\op{B}_1 \sigma_\op{x})
- q (\op{A}_1 \otimes \op{B}_1)}{\psi}, 
\end{equation}
allowing us to apply Csirelson's conversion between observables and
vectors~\cite{Cirel87, Cirel93}, as done in Wehner~\cite{Wehner05d}.
\begin{lemma}[Tsirelson~\cite{Cirel87, Cirel93}]
\label{lm:tsirelson}
Let $\op{A}_0, \ldots, \op{A}_{m-1}$ and $\op{B}_0, \ldots,
\op{B}_{m-1}$ be observables with eigenvalues in the interval $[-1,
  1]$.  Then for any state $\ket{\psi}$ shared between Alice and Bob,
there exist real unit vectors $x_0, \ldots, x_{m-1}$ and $y_0, \ldots,
y_{m-1}$ such that
\begin{equation}
\label{eq:tsirelsonconversion}
\tinner{\psi}{\op{A}_i \otimes \op{B}_j}{\psi} = x_i \cdot y_j
\end{equation}
for all $0 \leqslant i, j < m$.  Conversely, for any set of real unit
vectors, $x_0, \ldots, x_{m-1}$ and $y_0, \ldots, y_{m-1}$, and any
maximally entangled state $\ket{\psi}$, there exist observables
$\op{A}_i$ and $\op{B}_j$ with eigenvalues $\pm 1$ such that
Eq.~\ref{eq:tsirelsonconversion} holds for all $0 \leqslant i, j < m$.
\end{lemma}
We define five vectors, one vector for each of Alice's two observables
$\op{A}_0$ and $\op{A}_1$, one for Bob's observable $\op{B}_0$, and
\emph{two} vectors for Bob's observable $\op{B}_1$,
\begin{align*}
x_0 &= (\op{A}_0 \otimes \id) \ket{\psi} &
y_0 &= (\id \otimes \op{B}_0) \ket{\psi} &
z_0 &= (\id \otimes \op{B}_1) \ket{\psi}\\
x_1 &= (\op{A}_1 \otimes \id) \ket{\psi} &
&& z_1 &= (\id \otimes (\sigma_\op{x} \op{B}_1 \sigma_\op{x})) \ket{\psi}.
\end{align*}
\begin{figure}
\centering
\scalebox{1}{\includegraphics{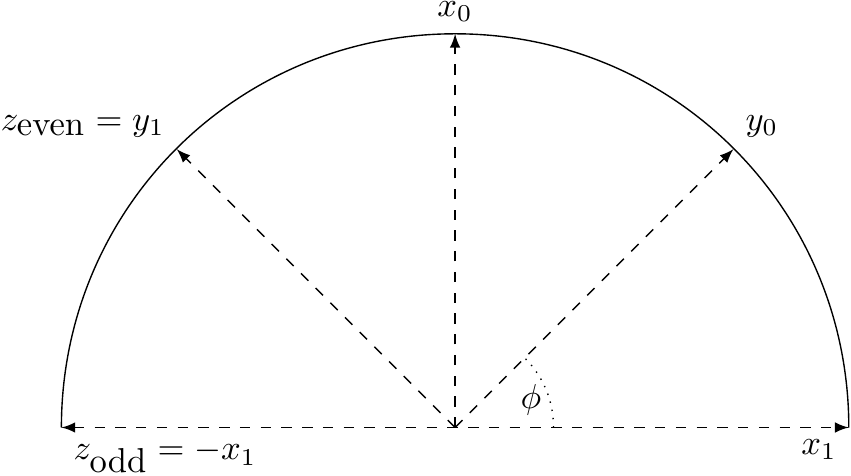}}
\caption{Geometric intuition for the choice of angle $\phi$, when \mbox{$0 < p \leqslant \frac{2}{3}$}. The choice determines the direction of measurements performed by Alice and Bob and how the vectors we obtain through Csirelson's vectorization relate to each other. For the single copy case $z_{\textrm{even}}$ and $z_{\textrm{odd}}$ correspond to the individual vectors $z_0$ and $z_1$ respectively. In the multiple copy case $z_{\textrm{even}}$ and $z_{\textrm{odd}}$ correspong to multiple vectors.}
\label{fig:qnlbfig}
\end{figure}
Let $G= [g_{ij}]$ be the Gram Matrix of the five vectors
$\{x_0,x_1,y_0,z_0,z_1\},$
\begin{equation*} 
G = 
\left(\begin{array}{ccccc}
x_0 \cdot x_0 & x_0 \cdot x_1 & x_0 \cdot y_0 & x_0 \cdot z_0 & x_0 \cdot z_1  \\
x_1 \cdot x_0 & x_1 \cdot x_1 & x_1 \cdot y_0 & x_1 \cdot z_0 & x_1 \cdot z_1 \\
y_0 \cdot x_0 & y_0 \cdot x_1 & y_0 \cdot y_0 & y_0 \cdot z_0 & y_0 \cdot z_1 \\
z_0 \cdot x_0 & z_0 \cdot x_1 & z_0 \cdot y_0 & z_0 \cdot z_0 & z_0 \cdot z_1 \\
z_1 \cdot x_0 & z_1 \cdot x_1 & z_1 \cdot y_0 & z_1 \cdot z_0 & z_1 \cdot z_1 
\end{array}\right),
\end{equation*}
and set $W$ to be the weight matrix
\begin{equation*}
W=
\begin{pmatrix}
 0 & \hphantom{-}0 & 1 & \hphantom{-}1 & \hphantom{-}0 \\
 0 & \hphantom{-}0 & 1 &-q &-p \\
 1 & \hphantom{-}1 & 0 & \hphantom{-}0 & \hphantom{-}0 \\
 1 &-q & 0 & \hphantom{-}0 & \hphantom{-}0 \\
 0 &-p & 0 & \hphantom{-}0 & \hphantom{-}0
\end{pmatrix}.
\end{equation*}
Finding an upper bound on the value~$V$ in Eq.~\ref{eq:optsinglecopy}
then becomes equivalent to finding an upper bound on the primal value
of the semidefinite program (SDP)
\begin{equation}
 \label{eq:sdp1}
\begin{split}
&\max\limits_{G} \frac{1}{2}\Tr(GW) \\
&\begin{split}
\text{ subject to } & G \succcurlyeq 0\\
& g_{ii} = 1  \text{ for all } i \in \{1,\ldots,5\}. 
\end{split}
\end{split}
\end{equation}
The constraint $G \succcurlyeq 0$ ensures that $G$ is a Gram matrix,
and the constraints that the diagonal entries of~$G$ are equal to~1,
ensure that the five vectors are of unit norm.  {From} any valid
solution to the primal, we can extract a set of \emph{five}
observables via Tsirelson's correspondence and construct a protocol
that has the same value as the primal solution, and vice-versa, from
any protocol, we can extract a set of five vectors, the Gram matrix of
which is a primal solution having the same value as the value attained
by the protocol.

We prove our upper bound on the primal value by giving a feasible
solution to the dual of value equal to the value $V$ in
Eq.~\ref{eq:valueofprotocolP}.  To~conclude that our dual solution is
feasible, we need to show that a particular matrix $\dualm$ is
positive semidefinite.  Rather than attempting conveying a technical
analysis of the roots of the matrix $\dualm$'s characteristic
polynomial, we shall instead break the matrix~$\dualm$ into smaller
parts and repeatedly apply the following simple observation about the
eigenvalues of a matrix of dimension \mbox{$2 \times 2$}.

\begin{observation}\label{obs:psd2by2}
A real-valued $2 \times 2$ matrix is \psd{} if and only if it has a
non-negative diagonal entry and its determinant is non-negative.
\end{observation}

To see this, notice that a symmetric real-valued matrix is \psd{} if
and only if one of its two eigenvalues is non-negative and the product
of its two eigenvalues is non-negative, which holds if and only if it
has a non-negative diagonal entry and its determinant is non-negative.

\begin{lemma} 
\label{lm:1cpyopt} 
The dual value of the SDP given in Eq.~\ref{eq:sdp1} is upper bounded
by the value attained by Protocol~\ref{protocolP} for \mbox{$n=1$}
given in Eq.~\ref{eq:valueofprotocolP}.
\end{lemma}

\begin{proof} 
Let $b = (1,1,1,1,1)$ be a vector in $\mathbb{R}^{5}$.  The dual of
the primal SDP in Eq.~\ref{eq:sdp1} is
\begin{equation}  \label{eq:sdp1_dual}
\begin{split}
&\min\limits_{\lambda} \lambda \cdot b^T\\
&\begin{split}
\text{ subject to } 
K = 2 \diag(\lambda) - W \succcurlyeq 0,
\end{split}
\end{split}
\end{equation}
where $\lambda$ is a vector in $\mathbb{R}^{5}$ and matrix
$\diag(\lambda)$ is of dimension $5 \times 5$ containing $\lambda_i$
in the $\nth{i}$ diagonal entry and zeroes off-diagonal.

First consider the range $\frac{2}{3} \leqslant p \leqslant 1$.  The
dual solution
\begin{equation*} 
\lambda = (1,p,1,\frac{p}{2},\frac{p}{2})
\end{equation*} 
has value $\lambda \cdot b^T = 2(1+p)$, matching the value of the
protocol given in Eq.~\ref{eq:valueofprotocolP}.  To show that the
constraint \mbox{$K \succcurlyeq 0$} for the dual problem is
satisfied, express matrix $K = 2\diag(\lambda) - W$ as the sum of two
matrices,
\begin{equation*} 
K = K_1 + K_2 = \left(\begin{array}{ccccc} 
\hphantom{-}2 & \hphantom{-}0 & -1 & -1 & 0 \\
\hphantom{-}0 & \hphantom{-}p & -1 & \hphantom{-}(1-p) & 0 \\
-1 & -1 & \hphantom{-}2 & \hphantom{-}0 & 0 \\
-1 & \hphantom{-}(1-p) & \hphantom{-}0 & \hphantom{-}p & 0 \\
\hphantom{-}0 & \hphantom{-}0 & \hphantom{-}0 & \hphantom{-}0 & 0 
\end{array}\right)+
\left(\begin{array}{ccccc} 
0 & 0 & 0 & 0 & 0 \\
0 & p & 0 & 0 & p \\
0 & 0 & 0 & 0 & 0 \\
0 & 0 & 0 & 0 & 0 \\
0 & p & 0 & 0 & p 
\end{array}\right).
\end{equation*} 
Matrix $K_2$ is a scaled projection with eigenvalues~$0$ and~$2p$ and
is therefore \psd.  For matrix~$K_1$, ignore its fifth row and column,
which are zero, and conjugate the remaining \mbox{$4 \times 4$}
submatrix of~$K_1$ by
$\frac{1}{\sqrt{2}}\left(\begin{smallmatrix}1&\hphantom{-}1\\1&-1
\end{smallmatrix}\right)\otimes \id$, yielding the submatrix
\begin{equation*} 
\left(\begin{array}{cccc} \hphantom{-}1 & -1 & 0 & 0 \\ -1 &
  \hphantom{-}1 & 0 & 0 \\ \hphantom{-}0 & \hphantom{-}0 & 3 & 1
  \\ \hphantom{-}0 & \hphantom{-}0 & 1 & 2p-1
\end{array}\right).
\end{equation*} 
The upper-left $2\times 2$ block is \psd, and, by
Observation~\ref{obs:psd2by2}, the lower-right block is \psd{} when
\mbox{$3(2p-1) \geqslant 1$}, which holds when \mbox{$p \geqslant
  \frac{2}{3}$}.  We have shown that matrix $K$ is the sum of two
\psd{} matrices, and it is therefore \psd.
 
Next consider the range $0 < p < \frac{2}{3}$.  The dual solution
\begin{equation*} 
\lambda = 
\cos(\phi)(1,q,1,q,0)+(0,\frac{p}{2},0,0,\frac{p}{2})
\end{equation*} 
has value $2(1+q)\cos(\phi) +p$, matching the value of the protocol
given in Eq.~\ref{eq:valueofprotocolP} for $n=1$. (When $n=1$, the
expression in Eq.~\ref{eq:valueofprotocolP} for the range $0< p <
\frac{1}{2}$ simplifies to the expression for the range $\frac{1}{2}
\leqslant p < \frac{2}{3}$.)  It remains to show that the constraint
\mbox{$K \succcurlyeq 0$} is satisfied.  Proceeding as in the case
$\frac{2}{3} \leqslant p \leqslant 1$, we write
\begin{equation*}
K=  K_1+K_2
=
\begin{pmatrix}
 \hphantom{-}2\cos(\phi) & \hphantom{-}0 & -1 & -1 & 0 \\
 \hphantom{-}0 & \hphantom{-}2q\cos(\phi) & -1 &\hphantom{-}q & 0 \\
 -1 & -1 & \hphantom{-}2\cos(\phi) & \hphantom{-}0 & 0 \\
 -1 & \hphantom{-}q & \hphantom{-}0 & \hphantom{-}2q\cos(\phi) & 0 \\
 \hphantom{-}0 & \hphantom{-}0 & \hphantom{-}0 & \hphantom{-}0 & 0
\end{pmatrix}
+
\begin{pmatrix}
 0 & 0 & 0 & 0 & 0 \\
 0 & p & 0 & 0 & p \\
 0 & 0 & 0 & 0 & 0 \\
 0 & 0 & 0 & 0 & 0 \\
 0 & p & 0 & 0 & p
\end{pmatrix}, 
\end{equation*}
and conjugate the upper-left $4 \times 4$ submatrix of~$K_1$ by
$\frac{1}{\sqrt{2}}\left(\begin{smallmatrix}1&\hphantom{-}1\\1&-1
\end{smallmatrix}\right)\otimes \id$, this time yielding the 
block matrix
\begin{equation*} 
\left(\begin{array}{cccc} 
 \hphantom{-}2\cos(\phi) - 1 & -1 & 0 & 0  \\
 -1 & \hphantom{-}q(2\cos(\phi)+1) & 0 & 0  \\
 \hphantom{-}0 & \hphantom{-}0 & 2\cos(\phi)+1 & 1  \\
 \hphantom{-}0 & \hphantom{-}0 & 1 & q(2\cos(\phi)-1) 
\end{array}\right).
\end{equation*} 
The two diagonal entries $2\cos(\phi)+1$ and $q(2\cos(\phi)+1)$ are
non-negative since both $\cos(\phi)$ and $q$ are non-negative.  Both
blocks have the same determinant $4 q \cos^2(\phi) - (1+q)$ which
equals zero.  (When $n=1$, the expression in Eq.~\ref{eq:phi} for the
range $0< p < \frac{1}{2}$ simplifies to the expression for the range
$\frac{1}{2} \leqslant p < \frac{2}{3}$.)  Applying
Observation~\ref{obs:psd2by2}, we conclude that $K$ is \psd.
\end{proof}

We have proved that the dual SDP is feasible and has a solution of
value no larger than the value attained by the protocol.  By
Tsirelson's correspondence, the protocol yields a feasible solution to
the primal SDP of the same value as the protocol.  These three values
must therefore be equal.  We~conclude that our protocol is optimal for
\mbox{$n=1$} and that the measurement angle specified by
Eq.~\ref{eq:phi} is optimal.

\section{Protocol~\ref{protocolP} is optimal for 2 and 3 copies}
\label{sec:2and3copies}

In~the preceding section, we show that no protocol can achieve a value
higher than our Protocol~\ref{protocolP} when given only a single copy
of a \bfunc{qNLB}.  We show that the same statement holds true for 2 and 3
copies of a \bfunc{qNLB}: Among all non-adaptive protocols for distillation
using at most 3 copies of a \bfunc{qNLB}, none attains a value strictly higher
than the value attained by our Protocol~\ref{protocolP} using the same
number of copies of a \bfunc{qNLB}.

\begin{theorem}\label{thm:optimalupto3copies}
Protocol~\ref{protocolP} is optimal among all non-adaptive protocols
using at most 3 copies of a \bfunc{qNLB}.
\end{theorem}

The proofs of the cases with multiple copies follow the outline we use
in the simple single-copy case, except that now some of the steps
become significantly more involved.  We first give a general
construction of a primal SDP and its dual SDP for any number $n$ of
copies of a \bfunc{qNLB}, stated as Eqs.~\ref{eq:generalPrimalSDP}
and~\ref{eq:generalDualSDP} in Appendix~\ref{appendix:sdpconstruction}
below.  The size of the dual SDP grows exponentially in the number of
copies $n$ utilized in the \bfunc{qNLB} protocol.  In Appendix~\ref{appendix:optimaldualsolution}
below, we give a complete analytical proof of its value when $n$ is at
most~$3$, proving Theorem~\ref{thm:optimalupto3copies}.

Our proof technique is general and should in principle be extendable to any \emph{fixed} higher value of~$n$.  More desirable, however, is to discover a method for analyzing our dual SDP for \emph{all} values of $n$ simultaneously.  It seems plausible that such a generic proof technique should exist, but finding one has thus far eluded us.  The dual SDP has an appealing representation in which we have been able to maintain many symmetries and letting it have an almost algorithmic structure. We~expect that the solution value obtained in Lemma~\ref{lm:dual3cpyTail} can be generalized to all values of~$n$. 

One possible route in proving a generalized version of Theorem~\ref{thm:optimalupto3copies} that holds for all values of~$n$ is to obtain a general form for the off-diagonal entries of the matrix~$W_\tail$.  If a proof of optimality for any $n$ could be found, it would imply that we could make a statement equally strong to the \bfunc{NLB} case, thus proving that our Protocol~\ref{protocolP} is optimal among all non-adaptive protocols for \bfunc{qNLB}s.

\section{Discussion}
\label{sec:discussion}

Our~\bfunc{qNLB} distillation protocol outperforms the optimal parity protocol for \bfunc{NLB}s due to the different route it takes to achieve distillation. The classical protocol determines the final output bits by computing the parity of the individual output from each box. For correlated \bfunc{NLB}s, this leaves the expectation value unchanged for inputs $00$, $01$ and $10$, while for input $11$ the expectation is increased for values of $p$ less than half. Our \bfunc{qNLB} protocol performs better than the classical optimal protocol by hedging. Non-adaptively, the players can tweak their measurements so as to set up stronger correlations. The combined reduction in the expectation values for inputs $00$, $01$ and $10$ is overwhelmed by the increase in the expected value for input $11$. Effectively, the entangled measurements on \bfunc{qNLB}s allow access to a set of strategies that are inaccessible to classical protocols for nonlocal boxes.

\begin{figure}[th]
\begin{center}
\scalebox{1}{\includegraphics{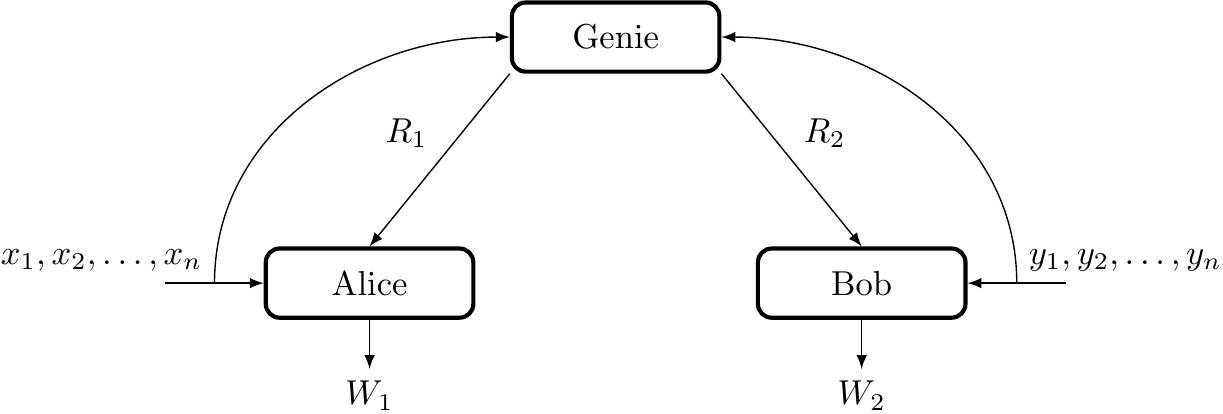}}
\end{center}
\caption[Assisted Common Information]{Assisted common information.}
\label{fig:assist}
\end{figure}
The framework we consider in Figure~\ref{fig:physmodel} is analogous to the idealised secure scenario for two-party computation model considered by Yao~\cite{Yao82}.  The trust assumptions for Charlie may be unrealistic for cryptographic protocols but the model highlights the notion that it is possible to consider \bfunc{NLB} correlations as a physical model rather than only as a hypothetical resource~\cite{Brassard05, Buhrman06}. Similar models have recently been considered under the notion of assisted common information (Figure~\ref{fig:assist}) within the cryptography community~\cite{Prabhak10}. The bounds obtained for nonlocality distillation may allow for improved cryptographic limits within an appropriate error model.

One~approach for obtaining an understanding of limits on quantum correlations is to develop underlying principles that are expected to be true for physical theories. The principle of information causality proposed by Paw{\l}owski et al.~\cite{Pawlowski09}, is one such principle. It states that the transmission of~$n$ classical bits can cause an information gain of at most~$n$ bits. The principle, a generalization of the no-signalling conditions, is violated by all correlations that violate Csirelson's bound. It is not known however, whether it is violated by all nonlocal correlations that are prohibited by quantum mechanics. Other physically motivated considerations include local quantum measurements due to Barnum et al.~\cite{Beigi09} and the uncertainty principle put forward by Oppenheim and Wehner~\cite{Wehner10}. 

In~our current work we show that if we restrict out attention to non-adaptive protocols, \bfunc{qNLB}s offer improved distillation over \bfunc{NLB}s.
A generalization of our SDP approach for \bfunc{qNLB}s may provide a similar result for adaptive protocols. This may imply distillability for correlations that are currently not known to be distillable and at the same time an increased understanding of correlations that violate principles such as information causality. Similarly, a new protocol for \bfunc{qNLB}s may close or reduce the trivial communication complexity gap~\cite{Brassard05}. At~the moment we do not have good insight into the structure of adpative protocols for \bfunc{qNLB}s or even a formulation that makes their analysis accessible.

The principle of macroscopic locality proposed by Navascu\'{e}s and Wunderlich~\cite{Navascues09}, states that a physical theory should recover classical physics in the macroscopic limit. In terms of nonlocal correlation this implies that as the number of particles with Alice and Bob become large, the joint distribution $p_{ab|xy}$ should admit a classical description. The principle suggests that quantum correlations identify exactly the set of correlations that are local {\em macroscopically}. The principle characterizes a slightly larger set, since Cavalcanti et al.~\cite{Cavalcanti10} showed that macroscopically local correlations can violate information causality. 

The principle of macroscopic locality may still identify exactly the set of quantum correlations, but require a stronger resource in form of a \bfunc{qNLB} to do so. Even if it turns out that our claim is invalid, the \bfunc{qNLB} model offers an alternate interpretation of macroscopic correlations that is not accessible with \bfunc{NLB}s. Apart from our obvious conclusion that \bfunc{qNLB}s offer stronger distillability than \bfunc{NLB}s, it is possible to interpret the distinction between quantum and classical attainable values in Figure~\ref{fig:distillsummary} as a separation, in principle, between quantum and classical predictions at the macroscopic level. The result corresponds to a form of Bell's inequality at the macroscopic level and implies a physical experimental framework within which this separation may be observed. Admittedly the result under consideration is restricted to non-adaptive distillation protocols and the Brunner and Skrzypczyk adaptive protocol is known to distill correlated \bfunc{NLB}s asymptotically to a perfect \bfunc{NLB}. Establishing optimal adaptive distillation protocols for \bfunc{qNLB}s that take into account coarse graining in measurements at the macroscopic level can demonstrate the following.
\begin{enumerate}
\item The principle of macroscopic locality identifies exactly the set of quantum correlations.

\item There exist correlations observable at the macroscopic level do not admit a classical description.
\end{enumerate}

In either case, we are led to a conclusion which improves our understanding of the correlations attainable in nature. A consequence of macroscopic quantum correlations is the possibility of identifying physical processes in nature that utilize these correlations at the biological level. Recent results regarding the inner workings of photosynthesis and bird navigation via Earth's magnetic field suggest that it very well may be so~\cite{Ball11}.

The main conclusion we draw from our work is that since \bfunc{NLB}s are not the strongest resource for producing no-signalling correlations, it is not appropriate to restrict attention to only this model when concerned with questions regarding ultimate limits on quantum nonlocality. Indeed, we prove that the \bfunc{qNLB} model offers stronger nonlocality distillability. Any~protocol for \bfunc{NLB}s can be simulated by \bfunc{qNLB}s, whereas \bfunc{NLB}s cannot simulate all \bfunc{qNLB} protocols. We~propose that questions that have so far been investigated within the framework of \bfunc{NLB}s should be re-evaluated using \bfunc{qNLB}s.

\section*{Acknowledgements} 
We~thank C\u{a}t\u{a}lin Dohotaru for useful discussions.  This work
was supported by the Canadian Institute for Advanced Research (CIFAR),
Canada's Natural Sciences and Engineering Research Council (NSERC),
the Canadian Network Centres of Excellence for Mathematics of
Information Technology and Complex Systems (MITACS), and QuantumWorks.

{
\small
\bibliographystyle{plainurl}
\bibliography{references}

\begin{thebibliography}{10}

\bibitem{Allcock09b}
J.~Allcock, N.~Brunner, N.~Linden, S.~Popescu, P.~Skrzypczyk, and
  T.~Vert\'{e}si.
\newblock Closed sets of nonlocal correlations.
\newblock {\em Physical Review A}, 80:062107, 2009.
\newblock \href {http://arxiv.org/abs/0908.1496} {\path{arXiv:0908.1496}},
  \href {http://dx.doi.org/10.1103/PhysRevA.80.062107}
  {\path{doi:10.1103/PhysRevA.80.062107}}.

\bibitem{Allcock09a}
J.~Allcock, N.~Brunner, M.~Paw{\l}owski, and V.~Scarani.
\newblock Recovering part of the boundary between quantum and nonquantum
  correlations from information causality.
\newblock {\em Physical Review A}, 80:040103, 2009.
\newblock \href {http://arxiv.org/abs/0906.3464v3} {\path{arXiv:0906.3464v3}},
  \href {http://dx.doi.org/10.1103/PhysRevA.80.040103}
  {\path{doi:10.1103/PhysRevA.80.040103}}.

\bibitem{Ball11}
P.~Ball.
\newblock Physics of life: {T}he dawn of quantum biology.
\newblock {\em Nature}, 474:272--274, 2011.
\newblock \href {http://dx.doi.org/10.1038/474272a}
  {\path{doi:10.1038/474272a}}.

\bibitem{Beigi09}
H.~Barnum, S.~Beigi, S.~Boixo, {M.\,B.} Elliott, and S.~Wehner.
\newblock Local quantum measurement and no-signaling imply quantum
  correlations.
\newblock {\em Physical Review Letters}, 104:140401, 2010.
\newblock \href {http://arxiv.org/abs/0910.3952v2} {\path{arXiv:0910.3952v2}},
  \href {http://dx.doi.org/10.1103/PhysRevLett.104.140401}
  {\path{doi:10.1103/PhysRevLett.104.140401}}.

\bibitem{Beckman01}
D.~Beckman, D.~Gottesman, {M.\,A.} Nielsen, and J.~Preskill.
\newblock Causal and localizable quantum operations.
\newblock {\em Physical Review A}, 64:052309, 2001.
\newblock \href {http://arxiv.org/abs/quant-ph/0102043}
  {\path{arXiv:quant-ph/0102043}}, \href
  {http://dx.doi.org/10.1103/PhysRevA.64.052309}
  {\path{doi:10.1103/PhysRevA.64.052309}}.

\bibitem{Benn03}
{C.\,H.} Bennett, {A.\,W.} Harrow, {D.\,W.} Leung, and {J.\,A.} Smolin.
\newblock On the capacities of bipartite hamiltonians and unitary gates.
\newblock {\em IEEE Transactions on Information Theory}, 49(9):1895--1911,
  2003.
\newblock \href {http://arxiv.org/abs/quant-ph/0205057v4}
  {\path{arXiv:quant-ph/0205057v4}}, \href
  {http://dx.doi.org/10.1109/TIT.2003.814935}
  {\path{doi:10.1109/TIT.2003.814935}}.

\bibitem{Brassard05}
G.~Brassard, H.~Buhrman, N.~Linden, {A.\,A.} M\'ethot, A.~Tapp, and F.~Unger.
\newblock A limit on nonlocality in any world in which communication complexity
  is trivial.
\newblock {\em Physical Review Letters}, 96:250401, 2006.
\newblock \href {http://arxiv.org/abs/quant-ph/0508042v1}
  {\path{arXiv:quant-ph/0508042v1}}, \href
  {http://dx.doi.org/10.1103/PhysRevLett.96.250401}
  {\path{doi:10.1103/PhysRevLett.96.250401}}.

\bibitem{Brunner10}
N.~Brunner, D.~Cavalcanti, A.~Salles, and P.~Skrzypczyk.
\newblock Bound nonlocality and activation.
\newblock {\em Physical Review Letters}, 106:020402, 2011.
\newblock \href {http://arxiv.org/abs/1009.4207v1} {\path{arXiv:1009.4207v1}},
  \href {http://dx.doi.org/10.1103/PhysRevLett.106.020402}
  {\path{doi:10.1103/PhysRevLett.106.020402}}.

\bibitem{Brunner09}
N.~Brunner and P.~Skrzypczyk.
\newblock Nonlocality distillation and post--quantum theories with trivial
  communication complexity.
\newblock {\em Physical Review Letters}, 102:160403, 2009.
\newblock \href {http://arxiv.org/abs/0901.4070} {\path{arXiv:0901.4070}},
  \href {http://dx.doi.org/10.1103/PhysRevLett.102.160403}
  {\path{doi:10.1103/PhysRevLett.102.160403}}.

\bibitem{Buhrman06}
H.~Buhrman, M.~Christandl, F.~Unger, S.~Wehner, and A.~Winter.
\newblock Implications of superstrong nonlocality for cryptography.
\newblock {\em Proceedings of the Royal Society A: Mathematical, Physical and
  Engineering Sciences}, 462:2071, 2006.
\newblock \href {http://arxiv.org/abs/quant-ph/0504133}
  {\path{arXiv:quant-ph/0504133}}, \href
  {http://dx.doi.org/10.1098/rspa.2006.1663}
  {\path{doi:10.1098/rspa.2006.1663}}.

\bibitem{Buhrman04}
H.~Buhrman and S.~Massar.
\newblock Causality and {C}irel'son bounds.
\newblock {\em Physical Review A}, 72:052103, 2005.
\newblock \href {http://arxiv.org/abs/quant-ph/0409066v2}
  {\path{arXiv:quant-ph/0409066v2}}, \href
  {http://dx.doi.org/10.1103/PhysRevA.72.052103}
  {\path{doi:10.1103/PhysRevA.72.052103}}.

\bibitem{Cavalcanti10}
D.~Cavalcanti, A.~Salles, and V.~Scarani.
\newblock Macroscopically local correlations can violate information causality.
\newblock {\em Nature Communications}, 1:136, 2010.
\newblock \href {http://arxiv.org/abs/1008.2624v2} {\path{arXiv:1008.2624v2}},
  \href {http://dx.doi.org/10.1038/ncomms1138} {\path{doi:10.1038/ncomms1138}}.

\bibitem{CHSH69}
{J.\,F.} Clauser, {M.\,A.} Horne, A.~Shimony, and {R.\,A.} Holt.
\newblock Proposed experiment to test local hidden--variable theories.
\newblock {\em Physical Review Letters}, 23:880, 1969.
\newblock \href {http://dx.doi.org/10.1103/PhysRevLett.23.880}
  {\path{doi:10.1103/PhysRevLett.23.880}}.

\bibitem{Cirel87}
{B.\,S.} Csirel'son.
\newblock Quantum analogues of {B}ell inequalities: {T}he case of two spatially
  separated domains.
\newblock {\em Journal of Soviet Mathematics}, 36(4):557--570, 1987.
\newblock \href {http://dx.doi.org/10.1007/BF01663472}
  {\path{doi:10.1007/BF01663472}}.

\bibitem{Cirel93}
{B.\,S.} Csirel'son.
\newblock Some results and problems on quantum {B}ell--type inequalities.
\newblock {\em Hadronic Journal Supplement}, 8(4):329--345, 1993.

\bibitem{Dam05}
{W.\,van} Dam.
\newblock Implausible consequences of superstrong nonlocality.
\newblock 2005.
\newblock \href {http://arxiv.org/abs/quant-ph/0501159}
  {\path{arXiv:quant-ph/0501159}}.

\bibitem{SWolf08a}
D.~Dukaric and S.~Wolf.
\newblock A limit on nonlocality distillation.
\newblock 2008.
\newblock \href {http://arxiv.org/abs/0808.3317} {\path{arXiv:0808.3317}}.

\bibitem{SWolf08b}
M.~Fitzi, E.~H\"{a}nggi, V.~Scarani, and S.~Wolf.
\newblock The nonlocality of $n$ noisy {P}opescu--{R}ohrlich boxes.
\newblock {\em Journal of Physics A: Mathematical and Theoretical},
  43(46):465305, 2010.
\newblock \href {http://arxiv.org/abs/0811.1649v2} {\path{arXiv:0811.1649v2}},
  \href {http://dx.doi.org/10.1088/1751-8113/43/46/465305}
  {\path{doi:10.1088/1751-8113/43/46/465305}}.

\bibitem{Forster11}
M.~Forster.
\newblock Bounds for nonlocality distillation protocols.
\newblock {\em Physical Review A}, 83:062114, 2011.
\newblock \href {http://arxiv.org/abs/0808.0651v4} {\path{arXiv:0808.0651v4}},
  \href {http://dx.doi.org/10.1103/PhysRevA.83.062114}
  {\path{doi:10.1103/PhysRevA.83.062114}}.

\bibitem{SWolf09}
M.~Forster, S.~Winkler, and S.~Wolf.
\newblock Distilling nonlocality.
\newblock {\em Physical Review Letters}, 102:120401, 2009.
\newblock \href {http://arxiv.org/abs/0809.3173v4} {\path{arXiv:0809.3173v4}},
  \href {http://dx.doi.org/10.1103/PhysRevLett.102.120401}
  {\path{doi:10.1103/PhysRevLett.102.120401}}.

\bibitem{Gus10}
G.~Gutoski.
\newblock {\em Quantum strategies and local operations}.
\newblock PhD thesis, University of Waterloo, 2010.
\newblock \href {http://arxiv.org/abs/1003.0038v1} {\path{arXiv:1003.0038v1}}.

\bibitem{Horn}
{R.\,A. } Horn and {C.\,R.} Johnson.
\newblock {\em Matrix analysis}.
\newblock Cambridge University Press, 1990.

\bibitem{Hoyer10}
P.~H{\o}yer and J.~Rashid.
\newblock Optimal protocols for nonlocality distillation.
\newblock {\em Physical Review A}, 82:042118, 2010.
\newblock \href {http://arxiv.org/abs/1009.1668v1} {\path{arXiv:1009.1668v1}},
  \href {http://dx.doi.org/10.1103/PhysRevA.82.042118}
  {\path{doi:10.1103/PhysRevA.82.042118}}.

\bibitem{Cirel85}
L.~Khalfin and {B.\,S.} Csirel'son.
\newblock Quantum and quasi-classical analogs of {B}ell inequalities.
\newblock {\em Symposium on the Foundations of Modern Physics}, page 441, 1985.

\bibitem{Kush}
E.~Kushilevitz and N.~Nisan.
\newblock {\em Communication {C}omplexity}.
\newblock Cambridge {U}niversity {P}ress, 1997.

\bibitem{Marcovitch07}
S.~Marcovitch, B.~Reznik, and L.~Vaidman.
\newblock Quantum--mechanical realization of a {P}opescu-{R}ohrlich box.
\newblock {\em Physical Review A}, 75:022102, 2007.
\newblock \href {http://arxiv.org/abs/quant-ph/0601122v4}
  {\path{arXiv:quant-ph/0601122v4}}, \href
  {http://dx.doi.org/10.1103/PhysRevA.75.022102}
  {\path{doi:10.1103/PhysRevA.75.022102}}.

\bibitem{Navascues09}
M.~Navascu\'{e}s and H.~Wunderlich.
\newblock A glance beyond the quantum model.
\newblock {\em Proceedings of the Royal Society A: Mathematical, Physical and
  Engineering Sciences}, 466:881--890, 2009.
\newblock \href {http://arxiv.org/abs/1107.3738v2} {\path{arXiv:1107.3738v2}},
  \href {http://dx.doi.org/10.1098/rspa.2009.0453}
  {\path{doi:10.1098/rspa.2009.0453}}.

\bibitem{Wehner10}
J.~Oppenheim and S.~Wehner.
\newblock The uncertainty principle determines the nonlocality of quantum
  mechanics.
\newblock {\em Science}, 330(6007):1072--1074, 2010.
\newblock \href {http://arxiv.org/abs/1004.2507v2} {\path{arXiv:1004.2507v2}},
  \href {http://dx.doi.org/10.1126/science.1192065}
  {\path{doi:10.1126/science.1192065}}.

\bibitem{Pawlowski09}
M.~Paw{\l}owski, T.~Paterek, D.~Kaszlikowski, V.~Scarani, A.~Winter, and
  M.~{\.Z}ukowski.
\newblock Information causality as a physical principle.
\newblock {\em Nature}, 461:1101--1104, 2009.
\newblock \href {http://arxiv.org/abs/0905.2292v3} {\path{arXiv:0905.2292v3}},
  \href {http://dx.doi.org/10.1038/nature08400}
  {\path{doi:10.1038/nature08400}}.

\bibitem{Piani06}
M.~Piani, M.~Horodecki, P.~Horodecki, and R.~Horodecki.
\newblock Properties of quantum nonsignaling boxes.
\newblock {\em Physical Review A}, 74:012305, 2006.
\newblock \href {http://arxiv.org/abs/quant-ph/0505110v1}
  {\path{arXiv:quant-ph/0505110v1}}, \href
  {http://dx.doi.org/10.1103/PhysRevA.74.012305}
  {\path{doi:10.1103/PhysRevA.74.012305}}.

\bibitem{Popescu94b}
S.~Popescu and D.~Rohrlich.
\newblock Quantum nonlocality as an axiom.
\newblock {\em Foundations of Physics}, 24(3):379--385, 1994.
\newblock \href {http://dx.doi.org/10.1007/BF02058098}
  {\path{doi:10.1007/BF02058098}}.

\bibitem{Prabhak10}
{V.\,M.} Prabhakaran and {M.\,M.} Prabhakaran.
\newblock Assisted common information with applications to secure two--party
  computation.
\newblock {\em IEEE International Symposium on Information Theory 2010}, pages
  2602--2606, 2010.
\newblock \href {http://arxiv.org/abs/1002.1916v1} {\path{arXiv:1002.1916v1}},
  \href {http://dx.doi.org/10.1109/ISIT.2010.5513743}
  {\path{doi:10.1109/ISIT.2010.5513743}}.

\bibitem{Short09}
{A.\,J.} Short.
\newblock No deterministic purification for two copies of a noisy entangled
  state.
\newblock {\em Physical Review Letters}, 102:180502, 2006.
\newblock \href {http://arxiv.org/abs/0809.2622v1} {\path{arXiv:0809.2622v1}},
  \href {http://dx.doi.org/10.1103/PhysRevLett.102.180502}
  {\path{doi:10.1103/PhysRevLett.102.180502}}.

\bibitem{Wehner05d}
S.~Wehner.
\newblock Tsirelson bounds for generalized
  {C}lauser--{H}orne--{S}himony--{H}olt inequalities.
\newblock {\em Physical Review A}, 73:022110, 2006.
\newblock \href {http://arxiv.org/abs/quant-ph/0510076v2}
  {\path{arXiv:quant-ph/0510076v2}}, \href
  {http://dx.doi.org/10.1103/PhysRevA.73.022110}
  {\path{doi:10.1103/PhysRevA.73.022110}}.

\bibitem{Yao82}
{A.\,C.} Yao.
\newblock Protocols for secure computations.
\newblock In {\em Proceedings of the 23rd Annual Symposium on Foundations of
  Computer Science}, pages 160--164, 1982.
\newblock \href {http://dx.doi.org/10.1109/SFCS.1982.88}
  {\path{doi:10.1109/SFCS.1982.88}}.

\end{thebibliography}
}
\appendix

\section{Proof of Lemma~\ref{lm:valueofprotocol}}
\label{appendix:protocolvalue}

In~this appendix, we prove that our Protocol~\ref{protocolP} given in
Section~\ref{sec:theprotocol} attains the value given by
Eq.~\ref{lm:valueofprotocol}.

\begin{lemma} \label{lm:ncpylemma}
For a mixed state $\rho = p\ket{\phi}\bra{\phi} +
q\ket{\psi}\bra{\psi}$, where $p \in [0,1]$ is a probability and $q =
1-p$ the complementary probability, the following trace relations
hold.
\begin{align*}
\Tr(\sigma_\op{z}^{\otimes n} \otimes \sigma_\op{z}^{\otimes n} \rho^{\otimes n}) &= (q-p)^n \\
\Tr(\sigma_\op{x}^{\otimes n} \otimes \sigma_\op{x}^{\otimes n} \rho^{\otimes n}) &= 1 \\
\Tr(\sigma_\op{z}^{\otimes n} \otimes \sigma_\op{x}^{\otimes n} \rho^{\otimes n}) 
&= \Tr(\sigma_\op{x}^{\otimes n} \otimes \sigma_\op{z}^{\otimes n} \rho^{\otimes n}) = 0.
\end{align*}
\end{lemma}

\begin{proof}
First consider the case $n=1$, 
\begin{align*}
\Tr(\sigma_\op{z} \otimes \sigma_\op{z} \rho) 
&=  q \tinner{\psi}{\sigma_\op{z} \otimes \sigma_\op{z}}{\psi}
  + p \tinner{\phi}{\sigma_\op{z} \otimes \sigma_\op{z}}{\phi}
  = q-p \\
\Tr(\sigma_\op{x} \otimes \sigma_\op{x} \rho) 
&=  q \tinner{\psi}{\sigma_\op{x} \otimes \sigma_\op{x}}{\psi}
  + p \tinner{\phi}{\sigma_\op{x} \otimes \sigma_\op{x}}{\phi} 
  = q+p = 1 \\
\Tr(\sigma_\op{x} \otimes \sigma_\op{z} \rho)
& = \Tr(\sigma_\op{z} \otimes \sigma_\op{x} \rho) 
  = q \tinner{\psi}{\sigma_\op{x} \otimes \sigma_\op{x}}{\psi}
  + p \tinner{\phi}{\sigma_\op{x} \otimes \sigma_\op{x}}{\phi} 
  = 0.
\end{align*}
For $n>1$, using that the operators are separable, rewrite
$\Tr(\sigma_1^{\otimes n} \otimes \sigma_2^{\otimes n} \rho^{\otimes
  n}) = \left( \Tr(\sigma_1 \otimes \sigma_2 \rho) \right)^n$ for all
Pauli operators $\sigma_1$ and~$\sigma_2$, and apply the
case~\mbox{$n=1$}.
\end{proof}

We~now prove Lemma~\ref{lm:valueofprotocol} by substituting the
appropriate expected values for each term in Equation~\ref{v:qnlb}.

\begin{proofof}{Lemma~\ref{lm:valueofprotocol}.}
Let Alice and Bob share~$n$ identical copies of a correlated \bfunc{qNLB} and
receive input bits~$x$ and~$y$ respectively. Application of
Protocol~\ref{protocolP} with observables~$\op{A}_x$ and~$\op{B}_y$
yields the following expectation values for inputs $00, 01$ and~$10$,
\begin{align}
\label{eq:eq1}
\bra{\psi}^{\otimes n} \op{A}_0 \otimes \op{B}_0 \ket{\psi}^{\otimes
  n} 
&=
\cos^2\left(\frac{\phi}{2}\right)-\sin^2\left(\frac{\phi}{2}\right) =
\cos\left(\phi\right), \\ 
\bra{\psi}^{\otimes n} \op{A}_0 \otimes
\op{B}_1 \ket{\psi}^{\otimes n} 
&=
 \cos\left(\frac{\phi}{2}\right)\cos\left(\frac{3\phi}{2}\right)
+\sin\left(\frac{\phi}{2}\right)\sin\left(\frac{3\phi}{2}\right)
= \cos\left(\phi\right) \nonumber \\ 
&= \bra{\psi}^{\otimes n}
\op{A}_1 \otimes \op{B}_0 \ket{\psi}^{\otimes n}. \nonumber
\end{align}
We reap the benefits of obtaining a lower value for the above inputs by obtaining a higher increase in the value for input~$11$ when \mbox{$0 < p < \frac{2}{3}$}. Applying Lemma~\ref{lm:ncpylemma}, the expectation value for input $11$ for \mbox{$0 < p < \frac{1}{2}$} is given by
\begin{align}
\label{eq:eq2}
\Tr\left(\op{A}_1 \otimes \op{B}_1 \rho^{\otimes n}\right)
&=\cos^2\left(\frac{3\phi}{2}\right)\Tr\left(\sigma_\op{z}^{\otimes n} \otimes \sigma_\op{z}^{\otimes n} \rho^{\otimes n}\right) -\sin^2\left(\frac{3\phi}{2}\right) \Tr\left(\sigma_\op{x}^{\otimes n} \otimes \sigma_\op{x}^{\otimes n} \rho^{\otimes n}\right) \nonumber \\
&=\left(q-p\right)^n\cos^2\left(\frac{3\phi}{2}\right)
-\sin^2\left(\frac{3\phi}{2}\right) \nonumber \\
&=\frac{1}{2}\left( (1+(q-p)^n)-(1-(q-p)^n) \right) 
\cos^2\left(\frac{3\phi}{2}\right) \nonumber \\
& -\frac{1}{2}\left( (1+(q-p)^n)+(1-(q-p)^n) \right) 
\sin^2\left(\frac{3\phi}{2}\right) \nonumber \\
&=\frac{1}{2}((1+(q-p)^n)\cos(3\phi)-(1-(q-p)^n)).
\end{align}
Equation~\ref{eq:eq2} yields the value for \mbox{$\frac{1}{2} \leqslant p \leqslant 1$} if we fix \mbox{$n=1$}. The expression is simplified by choosing $\phi$ as specified in Protocol~\ref{protocolP} and applying the following trignometric equivalence,
\begin{equation}
\label{eq:phi2}
\cos(3\phi) = \left\{ 
\begin{array}{ll}
\left( \frac{-2(q-p)^n}{1+(q-p)^n}\right) \cos(\phi) 
  & \text{ if } 0 < p < \frac{1}{2} \\
\frac{p-q}{q}\cos(\phi) 
  & \text{ if } \frac{1}{2} \leqslant p < \frac{2}{3}  \\
1 & \text{ if } \frac{2}{3} \leqslant p \leqslant 1.  \\
\end{array} \right.
\end{equation}
The value $V$ attained by Protocol~\ref{protocolP} is obtained by substituting the expectation values~\ref{eq:eq1} and~\ref{eq:eq2} in Equation~\ref{v:qnlb}, which gives
\begin{align*}
V &= \bra{\psi}^{\otimes n}( \op{A}_0 \otimes \op{B}_0 + \op{A}_0 \otimes \op{B}_1
+ \op{A}_1 \otimes \op{B}_0) \ket{\psi}^{\otimes n} - \Tr(\op{A}_1 \otimes \op{B}_1 \rho^{\otimes n}) \\
& = 3\cos\left(\phi\right)-\frac{1}{2}((1+(q-p)^n)\cos(3\phi)-(1-(q-p)^n)).
\end{align*}
To~complete the proof substitute Equation~\ref{eq:phi2} in the expression for~$V$ and simplify to obtain
\begin{equation*}
V =
\begin{cases}
(3+(q-p)^n)\cos(\phi) +\frac{1}{2}(1-(q-p)^n)
& \text{ if $0 < p < \frac{1}{2}$}\\
2(1+q)\cos(\phi) +p 
& \text{ if $\frac{1}{2} \leqslant p < \frac{2}{3}$}\\
2(1+p) 
& \text{ if $\frac{2}{3} \leqslant p \leqslant 1$.}
\end{cases}
\end{equation*}
\end{proofof}
The value attained by Protocol~\ref{protocolP} when Alice and Bob share~$n$ identical copies of a correlated \bfunc{qNLB} is strictly greater than the value attained by the optimal classical protocol for \mbox{$0 < p \leqslant \frac{1}{2}$}. To~verify the claim we need to show that the following inequality holds for \mbox{$0 < p \leqslant~\frac{1}{2}$},
\begin{equation*}
3-l < \frac{3+l}{2}\sqrt{\frac{3+l}{1+l}}+\frac{1-l}{2},
\end{equation*}
where \mbox{$l = (q-p)^n$} with~$l$ ranging between \mbox{$0 \leqslant l < 1$}. The inequality may be simplified to obtain,
\begin{equation*}
4(1+l) < (3+l)\sqrt{3+l}.
\end{equation*}
If~we substitute \mbox{$k=3+l$}, with $k$ ranging between \mbox{$3 \leqslant k < 4$}, we obtain the inequality,
\begin{equation*}
4k-k\sqrt{k}-8 < 0.
\end{equation*}
The inequality is verfied by checking that the expression on the left
hand side is negative for \mbox{$0<k<4$} and has roots at~$k$ equal
to~$4$. The limit of $(q-p)^n$ as $n$ approaches infinity is $0$ for
\mbox{$0 < p \leqslant \frac{1}{2}$}. We~conclude that
Protocol~\ref{protocolP} asymptotically distills correlated \bfunc{qNLB}s to
the value \mbox{$\frac{1}{2}(3\sqrt{3}+1) \approx 3.098076$} for~$p$
less than a half.

\section{Constructing the 
\texorpdfstring{$\boldsymbol{n}$}{\textbf\protect\textit{n}} copy SDP}
\label{appendix:sdpconstruction}

Let Alice and Bob share~$n$ identical copies of a correlated \bfunc{qNLB} and
receive input bits~$x$ and~$y$ respectively. Alice and Bob apply the
observables $\op{A}_x$ and $\op{B}_y$ respectively, as specified in
Protocol~\ref{protocolP}. Recall that the value attained for the CHSH
inequality is
\begin{equation*}
V = \bra{\psi}^{\otimes n}( \op{A}_0 \otimes \op{B}_0 + \op{A}_0 \otimes \op{B}_1
+ \op{A}_1 \otimes \op{B}_0) \ket{\psi}^{\otimes n} - \Tr(\op{A}_1 \otimes \op{B}_1 \rho^{\otimes n}).
\end{equation*}
We define \mbox{$N=2^n+3$} vectors, one vector for each of Alice's two observables $\op{A}_x$, one for Bob's observable~$\op{B}_0$, and \emph{$2^n$} vectors for Bob's observable $\op{B}_1$. Let the $2^n$ vectors $z_s$ be indexed by a length~$n$ bit string $s$ in $\{0,1\}^n$ and define \mbox{$\op{X}_s = \sigma_\op{x}^{s_1} \otimes \sigma_\op{x}^{s_2} \otimes \cdots \otimes \sigma_\op{x}^{s_n}$} so that,
\begin{align*}
x_0 &= (\op{A}_0 \otimes \id^{\otimes n}) \ket{\psi}^{\otimes n}\\
x_1 &= (\op{A}_1 \otimes \id^{\otimes n}) \ket{\psi}^{\otimes n}\\
y_0 &= (\id^{\otimes n} \otimes \op{B}_0) \ket{\psi}^{\otimes n}\\
z_s &= z_{s_1s_2 \ldots s_n} = (\id^{\otimes n} \otimes (\op{X}_s \op{B}_1 \op{X}_s)) \ket{\psi}^{\otimes n}.
\end{align*}
Let $G= [g_{ij}]$ be the Gram Matrix of the $N$ vectors \mbox{$\{x_0, x_1, y_0, z_{0^n},z_{0^{n-1}1},\ldots, z_{1^n}\}$}. Set $W$ to be the symmetric weight matrix
\begin{equation*}
w_{ij}=w_{ji} =
\begin{cases}
1
& \text{ if $(i,j)\in \{(1,3),(1,4),(2,3)\}$}\\
-q^{n-\abs{s}}p^{\abs{s}}
& \text{ if $i=2$ and $j=3+\abs{s}$}\\
0
& \text{otherwise,}
\end{cases}
\end{equation*}
where $\abs{s}$ is the Hamming weight of the bit string~$s$. Let $s,s',t$ and $t'$ be length $n$ bit strings in  $\{0,1\}^n$ such that $s \neq s'$ and $t \neq t'$. Optimizing the value $V$ attained by Protocol~\ref{protocolP} is then equivalent to finding an optimal primal solution to the following SDP.
\begin{equation}\label{eq:generalPrimalSDP}
\begin{split}
&\max\limits_{G} \frac{1}{2}\Tr(GW) \\
&\begin{split}
\text{ subject to } & G \succcurlyeq 0\\
& g_{ii} = 1  \text{ for all } i \in \{1,\ldots,N\}\\
& g_{3+\abs{s},3+\abs{s'}} = g_{3+\abs{t},3+\abs{t'}} \textrm{ if and only if } s \oplus s' = t \oplus t'.
\end{split}
\end{split}
\end{equation}
We~already encountered the first two set of constraints in the primal for the single copy SDP in Section~\ref{sec:singlecopy}. These constraints ensure that the matrix $G$ is a Gram matrix and the $N$ vectors used in its construction have unit norm. The new set of constraints are derived from the $(2^{n-1}-1)(2^n-1)$ inner product restrictions of the form $z_s \cdot z_{s'} = z_t \cdot z_{t'}$ on the $z_s$ vectors.

To~obtain the dual, we define vector $\lambda'$ in $\mathbb{R}^{M}$, with \mbox{$M=2^n+3+(2^{n-1}-1)(2^n-1)$}, where the first \mbox{$N=2^n+3$} components contribute to the solution value of the dual and the remaining entries correspond to the additional constraints. To distinguish between these two different roles we partition $\lambda'$ into two component vectors $\mu$ and $\tau$ such that,
\begin{equation*}
\mu_i = \left\{\begin{array}{ll}
\lambda_i'
   & \text{ if } 1 \leqslant i \leqslant N \\
0
   & \text{ if } N < i \leqslant M 
\end{array} \right. 
\quad \text{and} \quad
\tau_i = \left\{\begin{array}{ll}
0
   & \text{ if } 1 \leqslant i \leqslant N \\
\lambda_i'
   & \text{ if } N < i \leqslant M.
\end{array} \right.
\end{equation*}

Define vector \mbox{$b \in \mathbb{R}^M$} and let \mbox{$b_i=1$} for \mbox{$i \leqslant N$} and~$0$ otherwise. Given four unique length $n$ bit strings $s,s',t$ and $t'$, for each constraint of the form $s \oplus s' = t \oplus t'$, define a matrix $H_k$ for $N < k \leqslant M$,
\begin{equation*}
h_{ij} =
\begin{cases}
\hphantom{-}1
& \text{ if $i = 3+\abs{s}, j=3+\abs{s'}$ such that  $s \oplus s' = t \oplus t'$}\\
-1
& \text{ if $i = 3+\abs{t}, j=3+\abs{t'}$ such that  $s \oplus s' = t \oplus t'$}\\
\hphantom{-}0
& \text{otherwise.}
\end{cases}
\end{equation*}
The Lagrangian for the problem is given~by
\begin{eqnarray*}
\mathcal{L}(G,\lambda',Z) &=& \frac{1}{2}\Tr(GW)+\Tr(ZG)+ \Tr(\diag(\mu)-\diag(\mu)G)-\sum_{k=N+1}^M \tau_k\Tr(H_k G) \\
&=& \lambda' \cdot b + \Tr\left( \left(\frac{1}{2}W+Z-\diag(\mu)-\sum_{k=N+1}^M \tau_k H_k \right) G \right),
\end{eqnarray*}
where $\lambda'$ and \mbox{$Z \succcurlyeq 0$} are the dual variables. The dual function is then given~by
\begin{displaymath}
g(\lambda', Z) = \sup\limits_{G} \mathcal{L}(G,\lambda',Z) = \left\{ \begin{array}{ll}
\lambda' \cdot b       & \textrm{if $\frac{1}{2}W+Z-\diag(\mu)-\sum_{k=N+1}^M \tau_k H_k=0$}\\
+\infty & \textrm{otherwise.}
\end{array} \right.
\end{displaymath}
The dual problem may be stated as follows,
\begin{equation*}
\begin{split}
&\min\limits_{\lambda'} \lambda' \cdot b \\
&\begin{split}
\text{ subject to } & \frac{1}{2}W+Z-\diag(\mu)-\sum_{k=N+1}^M \tau_k H_k=0 \\ 
& Z \succcurlyeq 0.
\end{split}
\end{split}
\end{equation*}
We~simplify the formulation by removing variable $Z$ and defining $\lambda = 2\lambda'$ to obtain,
\begin{equation}
\label{eq:generalDualSDP}
\begin{split}
&\min\limits_{\lambda'} \lambda' \cdot b \\
&\text{ subject to } K  = 2\Big( \diag(\mu)-\sum_{k=N+1}^M \tau_k H_k \Big)-W \succcurlyeq 0.
\end{split}
\end{equation}
Oppenheim and Wehner~\cite{Wehner10} have recently shown that the strength of nonlocality is related to the uncertainty principle and entanglement steering. The off-diagonal constraints in our SDP formulation and the value they take may reveal additional insights about this relationship. Violation of one of the constraints may imply violation of a linked uncertainty relation. In~essence the constraints form a restriction on entanglement steering, where given Alice's measurements they restrict the states that Bob may now prepare and vice-versa.

In~Appendix~\ref{appendix:optimaldualsolution} we utilize the above formulation of the dual
to show that Protocol~\ref{protocolP} is the optimal non-adaptive
protocol for Alice and Bob when they have access to $2$ or $3$ copies
of correlated \bfunc{qNLB}s.

\section{Optimal dual solutions for \texorpdfstring{$\boldsymbol{2}$}{\textbf\protect\textit{2}} and \texorpdfstring{$\boldsymbol{3}$}{\textbf\protect\textit{3}} copies}
\label{appendix:optimaldualsolution}
The main idea we use to show optimality for the $2$ and $3$ copy
cases, as in the single copy case is to break up the constraint matrix
$K$ into a sum of matrices \mbox{$K = W_\head + W_\tail$} and show
that each matrix is \psd{}. We~decompose~$K$ such that there is a
fixed size $4 \times 4$ matrix $W_\head$, while the matrix $W_\tail$
has size $(2^n+1) \times (2^n+1)$. We~begin by defining the a cut-off
value~$x$ that determines the decomposition of $K$ into $W_\head$
and~$W_\tail$.
\begin{equation}
\label{eq:defnx}
x = \left\{ \begin{array}{ll}
\frac{1}{2}(1+(q-p)^n) & \text{ if } 0 < p \leqslant \frac{1}{2} \\
1-p & \text{ if } \frac{1}{2} < p < 1.
\end{array} \right.
\end{equation}
Define the matrix,
\begin{equation}
W_\head = 
\begin{pmatrix} 
\hphantom{-}\lambda_1 & \hphantom{-}0 & -1 & -1  \\
\hphantom{-}0 & \hphantom{-}l_1 & -1 & \hphantom{-}x \\
-1 & -1 & \hphantom{-}\lambda_3 & \hphantom{-}0  \\
-1 & \hphantom{-}x & \hphantom{-}0 & \hphantom{-}l_2 
\end{pmatrix}.
\end{equation}
The diagonal entries $\lambda_1$ and $\lambda_3$ are exactly the first and third components of the dual solution vector $\lambda$, while the entries $l_1$ and $l_2$ only have a partial contribution to the entries $\lambda_2$ and~$\lambda_4$. Next we determine the diagonal values of $W_\head$ and show that the matrix is \psd{} for these values.
\begin{lemma}
\label{lm:dual3cpyHead}
The dual value for matrix~$W_\head$ is given by 
\begin{equation*}
V' = \left\{\begin{array}{ll}
\sqrt{\frac{(1+x)^3}{x}} & \text{ if } 0 < p < \frac{2}{3} \\
3-x & \text{ if } \frac{2}{3} \leqslant p \leqslant 1.
\end{array} \right.
\end{equation*}
\end{lemma}
\begin{proof}
We~fix \mbox{$\lambda_1 = \lambda_3$} and \mbox{$l_1 = l_2$} and choose the diagonal entries as follows,
\begin{equation*}
\lambda_1= \left\{\begin{array}{ll}
\sqrt{1+\frac{1}{x}}
   & \text{ if } 0 < p < \frac{2}{3} \\
2   & \text{ if } \frac{2}{3}  \leqslant p \leqslant 1 
\end{array} \right. 
\quad \text{and} \quad
l_1 = \left\{\begin{array}{ll}
x\lambda_1     & \text{ if } 0 < p < \frac{2}{3} \\
1-x     & \text{ if } \frac{2}{3}  \leqslant p \leqslant 1.
\end{array} \right.
\end{equation*}
The dual value $V'$ achieved by~$W_\head$ can be calculated by summing up the diagonal entries. For $0 < p < \frac{2}{3}$, we obtain 
\begin{equation*}
V' = \lambda_1 + x\lambda_1 =(1+x)\sqrt{\frac{(1+x)}{x}} = \sqrt{\frac{(1+x)^3}{x}},
\end{equation*}
as required by Lemma~\ref{lm:dual3cpyHead}. Similarly, for $\frac{2}{3} \leqslant p \leqslant 1$, sum of the diagonal entries in~$W_\head$ is given by 
\begin{equation*}
V' = 2+1-x = 3-x.
\end{equation*} 
To~prove that the matrix is \psd{} we conjugate the matrix~$W_\head$ by \mbox{$\op{H} \otimes \id$},
\begin{equation*}
(\op{H} \otimes \id) W_\head (\op{H} \otimes \id) = 
\begin{pmatrix} 
\hphantom{-}\lambda_1 -1 & -1 & 0 & 0  \\
-1 & \hphantom{-}l_1+x & 0 & 0 \\
\hphantom{-}0 & \hphantom{-}0 & \lambda_1 +1 & 1  \\
\hphantom{-}0 & \hphantom{-}0 & 1 & l_1 - x 
\end{pmatrix}.
\end{equation*}
The matrix is \psd{} if \mbox{$(\lambda_1-1)(l_1+x) \geqslant 1$} and
\mbox{$(\lambda_1+1)(l_1 -x) \geqslant 1$}. Since both these inequalities hold for our choice of $\lambda_1$ and $l_1$, the matrix $W_\head$ is \psd{}.
\end{proof}
Unfortunately, we do not obtain a fixed size matrix $W_\tail$ similar
to~$W_\head$ that works for all~$n$. To~provide an overview of the dual constraints involved, we begin by giving a detailed construction of the $W_\tail$ matrix for the $3$ copy case. Let \mbox{$\ket{\Lambda} =  \ket{\psi}^{\otimes 3}$}. We~define the vectors $z_s$ as follows,
\begin{align*}
z_{000} & =  \id^{\otimes 3} \otimes \op{B}_1 \ket{\Lambda} \\
z_{001} & =  \id^{\otimes 3} \otimes (\id \otimes \id \otimes \sigma_\op{x})\op{B}_1(\id \otimes \id \otimes \sigma_\op{x}) \ket{\Lambda} \\
z_{010} & =  \id^{\otimes 3} \otimes (\id \otimes \sigma_\op{x} \otimes \id)\op{B}_1(\id \otimes \sigma_\op{x} \otimes \id) \ket{\Lambda} \\
z_{001} & =  \id^{\otimes 3} \otimes (\id \otimes \sigma_\op{x} \otimes \sigma_\op{x})\op{B}_1(\id \otimes \sigma_\op{x} \otimes \sigma_\op{x}) \ket{\Lambda} \\
z_{100} & =  \id^{\otimes 3} \otimes (\sigma_\op{x} \otimes \id \otimes \id)\op{B}_1(\sigma_\op{x} \otimes \id \otimes \id) \ket{\Lambda} \\
z_{101} & =  \id^{\otimes 3} \otimes (\sigma_\op{x} \otimes \id \otimes \sigma_\op{x})\op{B}_1(\sigma_\op{x} \otimes \id \otimes \sigma_\op{x})\ket{\Lambda} \\
z_{110} & =  \id^{\otimes 3} \otimes (\sigma_\op{x} \otimes \sigma_\op{x} \otimes \id)\op{B}_1(\sigma_\op{x} \otimes \sigma_\op{x} \otimes \id) \ket{\Lambda} \\
z_{111} & =  \id^{\otimes 3} \otimes (\sigma_\op{x} \otimes \sigma_\op{x} \otimes \sigma_\op{x})\op{B}_1(\sigma_\op{x} \otimes \sigma_\op{x} \otimes \sigma_\op{x}) \ket{\Lambda}.
\end{align*}
The following inner product constraints apply on the vectors $z_s$ due to their definition. These are exactly the additional constraints
required for the~$3$ copy dual solution.
\begin{alignat*}{6}
z_{001} \cdot z_{000} &=& \,\, z_{110} \cdot z_{111} &=& \,\, z_{011} \cdot z_{010} &=& \,\, z_{100} \cdot z_{101} \\
z_{010} \cdot z_{000} &=& z_{101} \cdot z_{111} &=& z_{100} \cdot z_{110} &=& z_{001} \cdot z_{011} \\
z_{011} \cdot z_{000} &=& z_{100} \cdot z_{111} &=& z_{001} \cdot z_{010} &=& z_{110} \cdot z_{101} \\
z_{100} \cdot z_{000} &=& z_{011} \cdot z_{111} &=& z_{110} \cdot z_{010} &=& z_{001} \cdot z_{101} \\
z_{101} \cdot z_{000} &=& z_{010} \cdot z_{111} &=& z_{001} \cdot z_{100} &=& z_{011} \cdot z_{110} \\
z_{110} \cdot z_{000} &=& z_{001} \cdot z_{111} &=& z_{100} \cdot z_{010} &=& z_{011} \cdot z_{101} \\
z_{111} \cdot z_{000} &=& z_{001} \cdot z_{110} &=& z_{101} \cdot z_{010} &=& z_{100} \cdot z_{011} 
\end{alignat*}
The dual constraint matrix~$K$ for the~$3$ copy case is given~by, 
\begin{equation*}
K=
\begin{pmatrix} 
\hphantom{-}\lambda_1 & \hphantom{-}0 & -1 & -1 & \hphantom{-}0 & \hphantom{-}0 & \hphantom{-}0 & \hphantom{-}0 & \hphantom{-}0 & \hphantom{-}0 & 0  \\
\hphantom{-}0 & \hphantom{-}\lambda_2 & -1 & \hphantom{-}q^3 & \hphantom{-}q^2p & \hphantom{-}q^2p & \hphantom{-}qp^2 & \hphantom{-}q^2p & \hphantom{-}qp^2 & \hphantom{-}qp^2 & p^3  \\
-1 & -1 & \hphantom{-}\lambda_3 & \hphantom{-}0 & \hphantom{-}0 & \hphantom{-}0 & \hphantom{-}0 & \hphantom{-}0 & \hphantom{-}0 & \hphantom{-}0 & 0  \\
-1 & \hphantom{-}q^3 & \hphantom{-}0 & \hphantom{-}\lambda_4 & \hphantom{-}0 & \hphantom{-}0 & \hphantom{-}0 & \hphantom{-}0 & \hphantom{-}0 & \hphantom{-}0 & 0  \\
\hphantom{-}0 & \hphantom{-}q^2p & \hphantom{-}0 & \hphantom{-}{\lambda_{12}+\lambda_{19}+\lambda_{26} } & \hphantom{-}\lambda_5 & \hphantom{-}0 & \hphantom{-}0 & \hphantom{-}0 & \hphantom{-}0 & \hphantom{-}0 & 0   \\
\hphantom{-}0& \hphantom{-}q^2p& \hphantom{-}0& \hphantom{-}{\lambda_{13}+\lambda_{20}+\lambda_{27} } & {-\lambda_{21}} & \hphantom{-}\lambda_6 & \hphantom{-}0 & \hphantom{-}0 & \hphantom{-}0 & \hphantom{-}0 & 0   \\
\hphantom{-}0& \hphantom{-}qp^2& \hphantom{-}0& \hphantom{-}{\lambda_{14}+\lambda_{21}+\lambda_{28}}& {-\lambda_{27}} & {-\lambda_{19}} & \hphantom{-}\lambda_7 & \hphantom{-}0 & \hphantom{-}0 & \hphantom{-}0 & 0   \\
\hphantom{-}0& \hphantom{-}q^2p& \hphantom{-}0& \hphantom{-}{\lambda_{15}+\lambda_{22}+\lambda_{29} }  & -\lambda_{23} & {-\lambda_{24}} & {-\lambda_{32}} & \hphantom{-}\lambda_8 & \hphantom{-}0 & \hphantom{-}0 & 0   \\
\hphantom{-}0& \hphantom{-}qp^2& \hphantom{-}0& \hphantom{-}\lambda_{16}+\lambda_{23}+\lambda_{30}          & {-\lambda_{29}} & {-\lambda_{25}} & {-\lambda_{31}} & {-\lambda_{26}} & \hphantom{-}\lambda_9 & \hphantom{-}0 & 0   \\
\hphantom{-}0& \hphantom{-}qp^2& \hphantom{-}0& \hphantom{-}{\lambda_{17}+\lambda_{24}+\lambda_{31} }  & {-\lambda_{18}} & {-\lambda_{22}} & -\lambda_{30} & {-\lambda_{20}} & {-\lambda_{28}} & \hphantom{-}\lambda_{10} & 0   \\
\hphantom{-}0&\hphantom{-}p^3& \hphantom{-}0&  \hphantom{-}{\lambda_{18}+\lambda_{25}+\lambda_{32} }  & {-\lambda_{17}} & {-\lambda_{16}} & {-\lambda_{15}} & {-\lambda_{14}} & {-\lambda_{13}} & {-\lambda_{12}} & \lambda_{11}   
\end{pmatrix},
\end{equation*}
where only the lower triangular matrix is shown for the
constraints. We~decompose~$K$ into~$W_\head$ which contains
contribution only from the upper left \mbox{$4 \times 4$} block matrix
with the remaining entries contained in~$W_\tail$.
\begin{align}
\label{Wtail3cpy}
K  = & W_\head + W_\tail \nonumber \\
 = &  
\begin{pmatrix} 
\hphantom{-}\lambda_1 & \hphantom{-}0 & -1 & -1  \\
\hphantom{-}0 & \hphantom{-}l_1 & -1 & \hphantom{-}x \\
-1 & -1 & \hphantom{-}\lambda_3 & \hphantom{-}0  \\
-1 & \hphantom{-}x & \hphantom{-}0 & \hphantom{-}l_2 
\end{pmatrix} \\
  &+
\begin{pmatrix} 
0 & 0 & 0 & 0 & \hphantom{-}0 & \hphantom{-}0 & \hphantom{-}0 & \hphantom{-}0 & \hphantom{-}0 & \hphantom{-}0 & 0  \\
0 & k_1 & 0 & q^3-x & \hphantom{-}q^2p & \hphantom{-}q^2p & \hphantom{-}qp^2 & \hphantom{-}q^2p & \hphantom{-}qp^2 & \hphantom{-}qp^2 & p^3  \\
0 & 0 & 0 & 0 & \hphantom{-}0 & \hphantom{-}0 & \hphantom{-}0 & \hphantom{-}0 & \hphantom{-}0 & \hphantom{-}0 & 0  \\
0 & q^3-x & 0 & k_2 & \hphantom{-}0 & \hphantom{-}0 & \hphantom{-}0 & \hphantom{-}0 & \hphantom{-}0 & \hphantom{-}0 & 0  \\
0&q^2p&0&{\lambda_{12}+\lambda_{19}+\lambda_{26} } & \lambda_5 & \hphantom{-}0 & \hphantom{-}0 & \hphantom{-}0 & \hphantom{-}0 & \hphantom{-}0 & 0   \\
0&q^2p&0&{\lambda_{13}+\lambda_{20}+\lambda_{27} } & {-\lambda_{21}} & \lambda_6 & \hphantom{-}0 & \hphantom{-}0 & \hphantom{-}0 & \hphantom{-}0 & 0   \\
0&qp^2&0&{\lambda_{14}+\lambda_{21}+\lambda_{28}}& {-\lambda_{27}} & {-\lambda_{19}} & \lambda_7 & \hphantom{-}0 & \hphantom{-}0 & \hphantom{-}0 & 0   \\
0&q^2p&0&{\lambda_{15}+\lambda_{22}+\lambda_{29} }  & -\lambda_{23} & {-\lambda_{24}} & {-\lambda_{32}} & \lambda_8 & \hphantom{-}0 & \hphantom{-}0 & 0   \\
0&qp^2&0&\lambda_{16}+\lambda_{23}+\lambda_{30}          & {-\lambda_{29}} & {-\lambda_{25}} & {-\lambda_{31}} & {-\lambda_{26}} & \lambda_9 & \hphantom{-}0 & 0   \\
0&qp^2&0&{\lambda_{17}+\lambda_{24}+\lambda_{31} }  & {-\lambda_{18}} & {-\lambda_{22}} & -\lambda_{30} & {-\lambda_{20}} & {-\lambda_{28}} & \lambda_{10} & 0   \\
0&p^3&0& {\lambda_{18}+\lambda_{25}+\lambda_{32} }  & {-\lambda_{17}} & {-\lambda_{16}} & {-\lambda_{15}} & {-\lambda_{14}} & {-\lambda_{13}} & {-\lambda_{12}} & \lambda_{11}   
\end{pmatrix}, \nonumber
\end{align}
where \mbox{$\lambda_2 = l_1+k_1$} and \mbox{$\lambda_4 = l_2+k_2$}. In~the following Lemma~\ref{lm:dual3cpyTail} we construct specific dual solution matrices that satisfy the constraint matrices of the above form for $2$ and $3$ copies of correlated \bfunc{qNLB}s. The proof of Lemma~\ref{lm:dual3cpyTail} for the case \mbox{$n=3$} utilizes the following generalization of Observation~\ref{obs:psd2by2}.
\begin{theorem}[Corollary 7.2.4 in~\cite{Horn}]
\label{thm:Hornpsd}
Let $A$ be a \mbox{$n \times n$} Hermitian matrix, and let
\begin{equation}
p_A(t) = t^n + a_{n-1}t^{n-1} + \cdots + a_{n-m}t^{n-m}
\end{equation}
be the characteristic polynomial of~$A$. Suppose that \mbox{$0
  \leqslant m \leqslant n$} and \mbox{$a_{n-m} \neq 0$}. Then $A$ is
\psd{} if and only if \mbox{$a_k \neq 0$} for all \mbox{$n-m \leqslant
  k \leqslant n$} and \mbox{$a_k a_{k+1} < 0$} for \mbox{$k = n-m,
  \ldots, n-1$}. We~define \mbox{$a_n \equiv~1$}.
\end{theorem}
Even though the formulation of Lemma~\ref{lm:dual3cpyTail} applies to the general $n$ copy case, we prove it only for $2$ and $3$ copy case due to the complexity of the off-diagonal constraints. 
\begin{lemma}
\label{lm:dual3cpyTail}
The matrix $W_\tail$ attains a dual solution value \, $1-x$, for \mbox{$n=2$} and \mbox{$n=3$}, where $x$ is defined in Equation~\ref{eq:defnx}.
\end{lemma}
\begin{proof}
First consider the case \mbox{$n=2$}, for the range \mbox{$\frac{1}{2} \leqslant p \leqslant 1$}. The matrix
\begin{align*}
W_\tail &= W_1-W_2 \\ \nonumber
&=p
\begin{pmatrix} 
1 & \hphantom{-}0 & \hphantom{-}0 & \hphantom{-}0 & \hphantom{-}0\\
0 & \hphantom{-}\frac{q}{2} & -\frac{q}{2} & -\frac{q}{2} & -\frac{q}{2}\\
0 & -\frac{q}{2} & \hphantom{-}\frac{q}{2} & \hphantom{-}\frac{q}{2} & \hphantom{-}\frac{q}{2}\\
0 & -\frac{q}{2} & \hphantom{-}\frac{q}{2} & \hphantom{-}\frac{q}{2} & \hphantom{-}\frac{q}{2}\\
0 & -\frac{q}{2} & \hphantom{-}\frac{q}{2} & \hphantom{-}\frac{q}{2} & \hphantom{-}p-\frac{q}{2}
\end{pmatrix}-
p
\begin{pmatrix} 
\hphantom{-}0 & q & -q & -q & -p \\
\hphantom{-}q & 0 & \hphantom{-}0 & \hphantom{-}0 & \hphantom{-}0  \\
-q & 0 & \hphantom{-}0 & \hphantom{-}0 & \hphantom{-}0 \\
-q & 0 & \hphantom{-}0 & \hphantom{-}0 & \hphantom{-}0\\
-p & 0 & \hphantom{-}0 & \hphantom{-}0 & \hphantom{-}0
\end{pmatrix},
\end{align*}
attains a dual solution value equal to \mbox{$1-x=p$}. It~remains to show that \mbox{$W_\tail \succcurlyeq 0$}. The matrix~$W_\tail$ has rank~$2$ and its row space is spanned by its first two rows.  The upper left \mbox{$2 \times 2$} submatrix $p \left(\begin{smallmatrix} \hphantom{-}1 & -q\\-q & \hphantom{-}\frac{q}{2}\end{smallmatrix}\right)$ is \psd{} by Observation~\ref{obs:psd2by2} since \mbox{$q \leqslant \frac{1}{2}$}.  It~follows that $W_\tail$ is \psd{} as well.

Next consider the range \mbox{$0 < p < \frac{1}{2}$}, for \mbox{$n=2$}. We define the matrix $W_\tail$ as
\begin{align*}
W_\tail &= W_1 - W_2 \\ \nonumber 
&= p
\begin{pmatrix} 
2q & \hphantom{-}0 & \hphantom{-}0 & \hphantom{-}0 & \hphantom{-}0\\
0 & \hphantom{-}\frac{q}{2} & -\frac{p}{2} & -\frac{p}{2} & -\frac{q}{2}\\
0 & -\frac{p}{2} & \hphantom{-}\frac{q}{2} & \hphantom{-}\frac{q}{2} & \hphantom{-}\frac{p}{2}\\
0 & -\frac{p}{2} & \hphantom{-}\frac{q}{2} & \hphantom{-}\frac{q}{2} & \hphantom{-}\frac{p}{2}\\
0 & -\frac{q}{2} & \hphantom{-}\frac{p}{2} & \hphantom{-}\frac{p}{2} & \hphantom{-}\frac{q}{2}
\end{pmatrix}-
p
\begin{pmatrix} 
\hphantom{-}0 & p & -q & -q & -p \\
\hphantom{-}p & 0 & \hphantom{-}0 & \hphantom{-}0 & \hphantom{-}0  \\
-q & 0 & \hphantom{-}0 & \hphantom{-}0 & \hphantom{-}0 \\
-q & 0 & \hphantom{-}0 & \hphantom{-}0 & \hphantom{-}0\\
-p & 0 & \hphantom{-}0 & \hphantom{-}0 & \hphantom{-}0\\
\end{pmatrix},
\end{align*}
which attains a dual solution value $2pq = 1-x$. The matrix has rank~$2$ and its row space is spanned by its
first two rows.  The upper left \mbox{$2 \times 2$} submatrix
$\left(\begin{smallmatrix} \hphantom{-}2q & -p\\-p &  \hphantom{-}\frac{q}{2}\end{smallmatrix}\right)$ is \psd{} by
Observation~\ref{obs:psd2by2} since \mbox{$p < q$}.  It~follows that
$W_\tail$ is \psd\ as well.

For the case \mbox{$n=3$}, we begin by removing the two zero rows and columns from the $W_\tail$ matrix as specified in Equation~\ref{Wtail3cpy}. We~further restrict columns of the same Hamming weight to be equal. This corresponds to equating the three cases each for weight~$1$ and~$2$ identified by the entries $q^2p$ and $qp^2$ respectively. The reduction in the number of constraints~$\lambda$ allows us to consider a \mbox{$5 \times 5 $} matrix. A~construction for the case \mbox{$p \leqslant \frac{1}{2}$} yields the matrix
\begin{equation}
\bordermatrix{%
&\hphantom{-}v_1 & \hphantom{-}v_2 & \hphantom{-}v_3 & \hphantom{-}v_4 & \hphantom{-}v_5 \\  \cr
&  \hphantom{-}p \left( 3q^2+p^2 \right) & -3qp^2 & \hphantom{-}q^2p & \hphantom{-}qp^2 & \hphantom{-}p^3 \\ \cr
&-3qp^2 & \hphantom{-}\frac{9}{2}\frac{q^2p^3}{q^2+p^2} & -\frac{3}{4}qp^2 & -\frac{3}{2}\frac{q^2p^3}{q^2+p^2} & -\frac{3}{4}qp^2 \\ \cr
&\hphantom{-}q^2p &-\frac{3}{4}qp^2 & \hphantom{-}\frac{1}{2}\frac {q^4p}{q^2+p^2} & \hphantom{-}\frac{1}{4}qp^2 & \hphantom{-}\frac{q^2p}{2} \frac{2
p^2-q^2 }{q^2+p^2} \\ \cr
&\hphantom{-}qp^2 & -\frac{3}{2}\frac{q^2p^3}{q^2+p^2} & \hphantom{-}\frac{1}{4} qp^2 & \hphantom{-}\frac{1}{2} \frac{q^2p^3}{q^2+p^2} & \hphantom{-}\frac{1}{4} qp^2 \\ \cr
&\hphantom{-}p^3 & -\frac{3}{4}qp^2 & \hphantom{-}\frac{q^2p}{2} \frac {  2p^2-q^2}{q^2+p^2} & \hphantom{-}\frac{1}{4} qp^2 & \hphantom{-}\frac{p}{2} \frac{ 3q^4-4q^2p^2+2p^4}{q^2+p^2} \\ \cr
}.
\end{equation}
The rank of~$W_\tail$ may be further decreased by noting that \mbox{$v_1 = 3v_3 + v_5$} and \mbox{$v_2 = -3v_4$}. We~now choose to consider the rank~$3$ matrix spanned by $v_1, v_3$ and $v_4$, given~by
\begin{equation}
\label{eq:tail1}
W_{\tail}=
\begin{pmatrix}
  p \left( 3q^2+p^2 \right) & q^2p & qp^2  \\ 
q^2p & \frac{1}{2}\frac {q^4p}{q^2+p^2} & \frac{1}{4}qp^2  \\ 
qp^2 & \frac{1}{4} qp^2 & \frac{1}{2} \frac{q^2p^3}{q^2+p^2}
\end{pmatrix}.
\end{equation}
The dual value for $W_{\tail}$ in Equation~\ref{eq:tail1} is given by $p \left( 3q^2+p^2
\right)$ which equals \mbox{$1-x$}.  A~similar construction for the
case \mbox{$\frac{1}{2} < p < 1$} yields the matrix
\begin{equation*}
\bordermatrix{%
&v_1 & \hphantom{-}v_2 & \hphantom{-}v_3 & \hphantom{-}v_4 & \hphantom{-}v_5 \\  \cr
&  p & \hphantom{-}q^3-q & \hphantom{-}q^2p & \hphantom{-}qp^2 & \hphantom{-}p^3 \\ \cr
&q^3-q & \hphantom{-}\frac{qp}{2}\frac{p^3-2p+2}{q^2+p^2} & -\frac{q^2p}{4}\frac{q^2+4qp+p^2}{q^2+p^2} & -\frac{qp}{2}\frac{q^3 + qp^2+p^3}{q^2+p^2} & -\frac{3}{4}q^2p \\ \cr
&q^2p &-\frac{q^2p}{4}\frac{q^2+4qp+p^2}{q^2+p^2} & \hphantom{-}\frac{1}{2}\frac {q^3p^2}{q^2+p^2} & \hphantom{-}\frac{1}{4}q^2p & \hphantom{-}\frac{q^2p}{2}\frac{q^2-qp+p^2}{q^2+p^2} \\ \cr
&qp^2 & -\frac{qp}{2}\frac{q^3 + qp^2+p^3}{q^2+p^2} & \hphantom{-}\frac{1}{4} q^2p & \hphantom{-}\frac{1}{2} \frac{qp^4}{q^2+p^2} & -\frac{q^2p}{4}\frac{q^2-4qp+p^2}{q^2+p^2}\\ 
\cr
&p^3 & -\frac{3}{4}q^2p & \hphantom{-}\frac{q^2p}{2}\frac{q^2-qp+p^2}{q^2+p^2} & -\frac{q^2p}{4}\frac{q^2-4qp+p^2}{q^2+p^2} & \hphantom{-}\frac{p^2}{2} \frac{7p^3-13p^2+11p-3}{q^2+p^2} \\ \cr
}.
\end{equation*}
The rank of~$W_\tail$ may be further decreased by noting that \mbox{$v_1 = 3v_3 + v_5$} and \mbox{$v_2 = -2v_3-v_4$}. We~now choose to consider the rank~$3$ matrix spanned by $v_1, v_3$ and $v_4$, given by
\begin{equation}
\label{eq:tail2}
W_{\tail}=
\begin{pmatrix}
  p  & q^2p & qp^2  \\ 
q^2p & \frac{1}{2}\frac {q^3p}{q^2+p^2} & \frac{1}{4}q^2p  \\ 
qp^2 & \frac{1}{4} q^2p & \frac{1}{2} \frac{qp^4}{q^2+p^2}
\end{pmatrix}.
\end{equation}
The dual value for $W_{\tail}$ in Equation~\ref{eq:tail2} is~$p$, which equals
\mbox{$1-x$}. The fact that the Matrices~\ref{eq:tail1}
and~\ref{eq:tail2} are \psd{} may be verified by application of
Theorem~\ref{thm:Hornpsd} to the characteristic polynomials of these
matrices.
\end{proof}
The final part of our analysis constitutes the proof of
Theorem~\ref{thm:optimalupto3copies} which is obtained by combining
the dual solution values for both the matrices $W_\head$ and
$W_\tail.$

\begin{proofof}{Theorem~\ref{thm:optimalupto3copies}.}
We~prove that the Protocol~\ref{protocolP} is optimal for $2$ and $3$
copies by combining the dual values from Lemmas~\ref{lm:dual3cpyHead}
and~\ref{lm:dual3cpyTail}. No~distillation for the range
\mbox{$\frac{1}{2} < p \leqslant 1$} implies the dual solution values
are the same for both \mbox{$n=2$} and \mbox{$n=3$}. Also, the value
attained matches the value attained by Protocol~\ref{protocolP} and is
therefore tight. For \mbox{$\frac{2}{3} < p \leqslant 1$}, we have
\begin{align*}
V & =  3 - x + 1 -x \\
& =  4-2x\\
& =  2(1+p).
\end{align*}
For \mbox{$\frac{1}{2} < p \leqslant \frac{2}{3}$}, we obtain
\begin{align*}
V & =  \sqrt{\frac{(1+x)^3}{x}} + 1 -x \\
& =  \sqrt{\frac{(2-p)^3}{1-p}} + p \\
& =  (2-p)\sqrt{\frac{2-p}{1-p}} + p \\ 
& =  2(2-p) \cos(\phi) + p \\ 
& = 3 \cos(\phi) - q\cos(3\phi) + p.
\end{align*}
For \mbox{$0 < p \leqslant \frac{1}{2}$} and \mbox{$n=2$},
\begin{align*}
V & =  \left( \frac{3+(q-p)^2}{2} \right)\sqrt{\frac{3+(q-p)^2}{1+(q-p)^2}} + \frac{1-(q-p)^2}{2} \\
& =  (3+(q-p)^2)\cos(\phi) +\frac{1}{2}(1-(q-p)^2) \\
& =  (3+(q-p)^2)\cos(\phi) +2pq.
\end{align*}
Finally, for \mbox{$0 < p \leqslant \frac{1}{2}$} and \mbox{$n=3$},
\begin{align*}
V & =  \left( \frac{3+(q-p)^3}{2} \right)\sqrt{\frac{3+(q-p)^3}{1+(q-p)^3}} + \frac{1-(q-p)^3}{2} \\
& =  (3+(q-p)^3)\cos(\phi) +\frac{1}{2}(1-(q-p)^3).
\end{align*}
\end{proofof}
This concludes our proof of Theorem~\ref{thm:optimalupto3copies} that establishes \bfunc{qNLB}s as a stronger resource for nonlocality for non-adaptive protocols. We~have shown that if we restrict out attention to non-adaptive protocols, \bfunc{qNLB}s offer improved distillation over \bfunc{NLB}s.
\end{document}